\documentclass[11pt,a4paper]{article}

\usepackage{amsmath}
\usepackage{amssymb}
\usepackage{amsthm}
\usepackage{cite}
\usepackage{hyperref}
\usepackage{subfigure}
\usepackage{algorithmic}
\usepackage{algorithm}
\usepackage[noshell]{gnuplottex}
\usepackage[usenames]{color}
\usepackage{pst-pdf}

\newcommand{\Ghz}{\texttt{Ghz}}

\newcommand{\GAR}{\operatorname{GAR}}          
\newcommand{\subGAR}{\operatorname{sub-GAR}}   
\newcommand{\FAR}{\operatorname{FAR}}          
\newcommand{\FTCR}{\operatorname{FTCR}}        
\newcommand{\subFTCR}{\operatorname{sub-FTCR}} 
\newcommand{\GDT}{\operatorname{GDT}}          
\newcommand{\IDT}{\operatorname{IDT}}          

\newcommand{\F}{{\bf F}}     
\newcommand{\G}{{\bf G}}     
\newcommand{\C}{{\bf C}}     
\newcommand{\V}{{\bf V}}     
\newcommand{\U}{{\bf U}}     
\newcommand{\A}{{\bf A}}     
\newcommand{\B}{{\bf B}}     
\newcommand{\E}{{\bf E}}     
\newcommand{\T}{T}           
\newcommand{\D}{\mathcal{D}} 

\newcommand{\tmin}{t_{\min}}
\newcommand{\tmax}{t_{\max}}
\newcommand{\Tgen}{T_{\operatorname{gen}}}
\newcommand{\Tchaff}{T_{\operatorname{chaff}}}
\newcommand{\Tvault}{T_{\operatorname{vault}}}

\newcommand{\minutia}{\mathfrak{m}}

\newcommand{\minutes}{~\textit{min}}
\newcommand{\seconds}{~\textit{sec}}
\newcommand{\hours}{~\textit{hours}}
\newcommand{\days}{~\textit{days}}
\newcommand{\months}{~\textit{months}}
\newcommand{\years}{~\textit{years}}

\newcommand{\quant}{\operatorname{quant}}

\newtheorem{thm}{Theorem}
\newtheorem{lemma}{Lemma}

\newtheorem{example}{Example}
\newtheorem{defn}{Definition}

\begin{document}

\title{Attacks and Countermeasures in Fingerprint\\Based Biometric Cryptosystems} 

\author{Benjamin Tams
\thanks{Benjamin Tams is with the Institute for Mathematical Stochastics, University of Goettingen, Goldschmidtstr. 7, 37077, Goettingen, Germany. Phone: +49-(0)551-3913515. Email: btams@math.uni-goettingen.de}
} 

\maketitle

\begin{abstract}
We investigate implementations of biometric cryptosystems protecting fingerprint templates (which are mostly based on the \emph{fuzzy vault scheme} by Juels and Sudan in 2002) with respect to the security they provide. We show that attacks taking advantage of the system's false acceptance rate, i.e. \emph{false-accept attacks}, pose a very serious risk --- even if \emph{brute-force attacks} are impractical to perform. Our observations lead to the clear conclusion that currently a single fingerprint is not sufficient to provide a secure biometric cryptosystem. But there remain other problems that can not be resolved by merely switching to multi-finger: Kholmatov and Yanikoglu in 2007 demonstrated that it is possible to break two matching vault records at quite a high rate via the \emph{correlation attack}.

We propose an implementation of a minutiae fuzzy vault that is inherently resistant against \emph{cross-matching} and the correlation attack. Surprisingly, achieving cross-matching resistance is not at the cost of authentication performance. In particular, we propose to use a randomized decoding procedure and find that it is possible to achieve a $\GAR=91\%$ at which no false accepts are observed on a database generally used. Our ideas can be adopted into an implementation of a multibiometric cryptosystem. All experiments described in this paper can fully be reproduced using software available for download.\footnote{\label{fn:thimble}Source code of the programs can be downloaded along with source code of a C++ library \texttt{thimble} from \url{http://www.stochastik.math.uni-goettingen.de/biometrics}.}
\end{abstract}

\section*{Keywords}
fingerprint, fuzzy vault, cryptanalysis, cross-matching, correlation attack

\section{Introduction}
In a traditional password-based authentication scheme, user names along with their respective passwords are stored on a server-side database. In such a scenario, we usually cannot prevent the fact that some persons will have access to the content of the database. Such persons (for example, system administrators) are thus able to read the password information related to enrolled users. To prevent them from using password information to impersonate authorized users, a non-invertible transformation\footnote{Non-invertible transformation of a password means that it is easy to transform the password, while on the contrary the derivation of a password from a given transformation is computationally
hard.} (e.g., a \emph{one-way hash function}) of each password is stored rather than the unprotected password. During authentication, the user sends his password to the service provider, which computes the password's non-invertible transformation. Next, the service provider compares the transformation that is stored in the database with the just-transformed password. If both agree, the authentication attempt is accepted; otherwise it is rejected. So a would-be thief cannot find the passwords in the database; he must either steal them from users or guess them.

However, a password can be forgotten or, if written down, can be stolen. To prevent these risks, many individuals attempt to create easily memorable passwords. Unfortunately, this often results in the individual choosing a weak password, typically constructed using personal information (e.g., a birthday or the name of a significant other), which increases the risk of others' guessing the password --- and therefore, of theft.

A popular alternative to password is to base authentication on biometrics such as fingerprints. Authentication protocols that incorporate biometric templates do not have the disadvantage that they can be forgotten or lost: Barring injuries, fingers are always with us; moreover, fingerprint features remain reasonably invariant over time.

If a biometric authentication scheme is in place the incorporated templates have especially to be stored protected. The situation is rather serious, because biometric templates may correspond to human beings with nearly unique precision compared to mere passwords. In addition, ineffective protection of biometric templates has consequences beyond breach of privacy, as compared to protecting passwords: For example, if a password is corrupted (e.g., discovered by others) it can be replaced easily compared to replacing a biometric template.

While the requirements for so-called \emph{biometric template protection schemes} are similar to those used for protecting passwords, they are more difficult to achieve: With high confidence it must be efficiently verifiable whether a provided biometric template matches the template that is encrypted by the stored data; furthermore, it must be computationally infeasible to derive the unencrypted biometric template from the stored data. There is one great difference between password and biometric authentication schemes: Contrary to passwords different measurements of the same biometric source will differ, while also having some reasonable similarity. These differences between two biometric templates of the same individual can be usefully conceptualized as deviations or errors. In this vein, there have been proposals for biometric template protections schemes that couple techniques from traditional cryptography with techniques from the discipline of \emph{error-correcting codes}.

\subsection{The Fuzzy Vault Scheme}

In 2002, Juels and Sudan proposed the \emph{fuzzy vault scheme} \cite{bib:JuelsSudan2002} which is a construction for protecting noisy data. While the \emph{fuzzy commitment scheme} \cite{bib:JuelsWattenberg1999}, which was proposed by Juels and Wattenberg in 1999, requires the data to be presented as a fixed-length feature vector, the fuzzy vault scheme allows the length of the data to vary and the data to be unordered. These properties enable the fuzzy vault scheme to protect fingerprint templates such as fingerprint minutiae. It works as follows.

\subsubsection*{Enrollment}
Given a fingerprint template containing $t$ fingerprint features, e.g., minutiae, its elements are encoded as elements $x$ in a fixed finite field $\F$. One chooses a secret message polynomial $f\in\F[X]$ in the indeterminate $X$ of degree smaller than $k$ and evaluates $f(x)$ at the encoded element $x$. The genuine pairs $(x,f(x))$ are dispersed among a large set of chaff points that do not lie on the graph of $f$, such that a vault of size $n$ is built.

\subsubsection*{Authentication}
Using a second genuine template one aims at distinguishing the genuine points from the chaff points. Given the points are mainly genuine, one can tolerate errors within certain limits determined by error-correcting codes \cite{bib:PetersonGorensteinZierler1960,bib:Berlekamp1966,bib:Massey1969,bib:Gao2002}.

\subsubsection*{Security}
From the difficulty of the problem of distinguishing genuine from chaff (without the help of a second genuine template) the fuzzy vault scheme draws its security. This problem can be reduced from the \emph{polynomial reconstruction problem} which is believed to be hard in general if $t\ll\sqrt{(k-1)\cdot n}$ \cite{bib:Sudan1997,bib:GuruswamiSudan1998,bib:BleichenbacherNguyen2000,bib:GuruswamiVardy2005,bib:KiayiasYung2008}.

\subsection{The Fuzzy Fingerprint Vault}
There are several biometric traits which can be used in biometric authentication systems (see \cite{bib:JainFlynnRoss2007} for an overview). Every biometric discipline is associated with very individual challenges which have to be solved before incorporating them into a biometric template protection scheme. One biometric trait that can be extracted from humans are his \emph{fingerprints} \cite{bib:Handbook2009}. This papers focuses on fingerprints and examines their adequacy of being protected by the fuzzy vault scheme. Even though there are other biometric template protection schemes \cite{bib:DodisEtAl2008} and implementations for fingerprints \cite{bib:ArakalaJeffersHoradam2007} the fuzzy vault scheme is the most dominant and promising scheme for which implementations to protect fingerprints have been proposed. 

\subsubsection*{Implementations}
Several fuzzy vault variants for protecting fingerprint minutiae templates can be found in the literature \cite{bib:ClancyKiyavashLin2003, bib:YangVerbaudwhede2005, bib:UludagPankantiJain2005, bib:UludagJain2006, bib:NandakumarJainPankanti2007, bib:NandakumarNagarJain2007,bib:Nagar2008,bib:Nagar2010}.

To increase vault practicability as well as security, Nagar et al. (2008, 2010) \cite{bib:Nagar2008,bib:Nagar2010} proposed to fuse a fingerprint's minutiae template with information about its ridge orientation and frequency by means of \emph{minutiae descriptors} \cite{bib:Feng2008}. In distinguishing genuine from chaff minutiae an attacker has in addition to guess the respective minutiae descriptors. This adds some security to the base vault implementation.

All the implementations above require an \emph{alignment step} where the query minutiae templates are aligned to the vault. This is very challenging since the enrolled templates are protected. Currently, the alignment is realized by techniques using auxiliary alignment data, e.g., see \cite{bib:UludagJain2006,bib:NandakumarJainPankanti2007}. The use of auxiliary alignment data, however, may cause security issues resulting in the leakage of information from the protected fingers.

Another interesting approach is to use \emph{alignment-free features}, i.e. features that do not dependent of the finger's rotation or displacement. Li et al. (2010) \cite{bib:LiEtAl2010} proposed to fuse \emph{minutiae local structures} \cite{bib:JiangYau2000} with \emph{minutiae descriptors} \cite{bib:Feng2008} and to protect them by the fuzzy vault scheme. The recognition performance that the authors report is promising and the error-prone step of aligning the query fingerprint to the vault is circumvented. Furthermore, the problem of information leakage by auxiliary alignment data does not exist anymore.

\subsection{Content and Contribution of the Paper}
After we described the functioning of a minutiae fuzzy vault in more detail (Section \ref{sec:MinutiaeFuzzyVault}) we investigate the security of implementations from the literature in different attack scenarios (Section \ref{sec:Attacks}). Reproducing the work of Mih\u{a}ilescu et al. (2009) \cite{bib:MihailescuMunkTams2009} we show that brute-force attacks can  be very practical to perform against most implementations (see Section \ref{sec:BFAttack}).

But even if brute-force attacks are infeasible to perform, there remains the possibility of the attacker to run an attack that takes it advantage out of the system's \emph{false acceptance rate}, i.e. \emph{false-accept attack}; also see Section 26.6.1.1 in \cite{bib:Biometrics2009}. We show that false-accept attacks are even much easier to perform than brute-force attacks (see Section \ref{sec:FARAttack}). Note that the false-accept attack is not restricted to the fuzzy vault scheme but it can be applied with virtually no modifications to every authentication scheme. Its attack success rate only depends on the system's false acceptance rate and the average time needed to run an impostor recognition attempt. Therefore, our observations clearly advocate that biometric cryptosystems merely based on a single finger cannot provide effective security. Rather multi-finger cryptosystems \cite{bib:MerkleEtAl2010b} (or even \emph{multibiometric cryptosystems} \cite{bib:NagarNandakumarJain2012}) should be developed.

For the fuzzy vault scheme there remains a problem that can not be solved merely by switching to multibiometrics. Given two matching instances of a minutiae fuzzy vault to an adversary he can correlate them; genuine minutiae tend to agree well in comparison to chaff minutiae, which are likely to be in disagreement. Thus, an intruder may reliably determine whether two vault records match, i.e. \emph{cross-matching}. Even worse, via correlation the adversary can try to distinguish genuine minutiae from chaff minutiae. If in this way a set of vault minutiae can be extracted that contains a reasonable proportion of genuine minutiae, then the vault can efficiently be broken. Consequently, this attack is called \emph{correlation attack}. Scheirer and Boult (2007) were the first who have drawn the attention to the risk of attacking fuzzy vault via record multiplicity \cite{bib:ScheirerBoult2007}. Then Kholmatov and Yanikoglu (2008) have demonstrated the practicability of the correlation attack \cite{bib:KholmatovYanikoglu2008}. Therefore, in Section \ref{sec:CMRFV}, we propose an implementation of a minutiae fuzzy vault that is inherently resistant against cross-matching and the correlation attack. Fortunately, cross-matching resistance can be achieved without decreasing the verification performance as we found in a test on a fingerprint database publicly available (see Section \ref{sec:CMRFVEvaluation}). This is mainly due to a randomized decoding procedure that we propose.

A final discussion, conclusion, and an outlook are given in Section \ref{sec:Discussion}.

All experiments described in this paper can fully be reproduced using software that we made available for download.\textsuperscript{\ref{fn:thimble}}

\section{Minutiae Fuzzy Vault Implementation}
\label{sec:MinutiaeFuzzyVault}
Assume tha we are given a minutiae template $\{(a,b,\theta)\}$ where $(a,b)$ and $\theta$ denote its position and angle, respectively. Using the fuzzy vault scheme we may protect the template as follows.

\subsection{Enrollment}
We describe the vault construction analogous to Nandakumar et al. (2007) \cite{bib:NandakumarJainPankanti2007} with some minor modifications. 

As in \cite{bib:NandakumarJainPankanti2007}, only well-separated minutiae are selected. Furthermore, only the $t\leq\tmax$ minutiae of best quality that are well-separated are selected for vault construction. If it is not possible to select at least a certain number of $\tmin$ minutiae, the enrollment is aborted and a \emph{failure to capture} is reported. Otherwise, the construction continues as follows. To hide the selected \emph{genuine minutiae} $\Tgen$, a set of \emph{chaff minutiae} $\Tchaff$ is generated at random fulfilling the following side conditions: First, each chaff minutia has the property that it is well-separated from all other vault minutiae --- genuine and chaff; second, a chaff minutia's position lays within the corresponding fingerprint image's region; third, the number of chaff minutiae is such that the vault minutiae reach a predefined size $n\geq t$, i.e. $n-t$ chaff minutiae are generated. The union of genuine and chaff minutiae is referred to as the \emph{vault minutiae} $\Tvault=\Tgen\cup\Tchaff$.

After the vault minutiae have been established, a secret is encoded as a polynomial $f$ of degree $<k$ having coefficients in a fixed finite field $\F=\mathbb{F}_q$ of size $q\geq n$. Now, list the vault minutiae as $(a_0,b_0,\theta_0),\hdots,(a_{n-1},b_{n-1},\theta_{n-1})$ by some convention, e.g., by sorting them in lexicographical order. By $x_0,\hdots,x_{n-1}\in\F$ denote $n$ distinct elements of the finite field. In this way, each list index $i=0,\hdots,n-1$ uniquely encodes an element in $\F$. Now we build the \emph{genuine set} as $\G =\{(x_i,f(x_i))~|~(a_i,b_i,\theta_i)\in \Tgen\}$.
Analogously, the \emph{chaff set} is defined as $\C=\{(x_j,y_j)~|~(a_j,b_j,\theta_j)\in \Tchaff\}$ where the $y_j$ s are chosen uniformly at random such that $y_j\neq f(x_j)$. The union $\V=\G\cup\C$ builds the \emph{vault points}.

The protected template is published as the triple $(\V,T_{\operatorname{vault}},h(f))$ where $h(f)$ denotes a cryptographic hash value of $f$ (e.g., SHA-1) to allow safe recovery of $f$ at genuine authentication.

Note, that there is a one-to-one correspondence between the vault $\V$ and the vault minutiae $\Tvault$. Thus, given a genuine minutia we also know its corresponding vault point and vice versa. In this respect, our construction is different from the construction of Nandakumar et al. (2007) who encode the minutiae information on the $x$-coordinate of its corresponding vault point. Another difference of our construction is that we use a SHA-1 hash value instead of constituting the secret polynomial with redundancy bits.

\begin{figure}[!ht]
\subfigure[]{\includegraphics[width=0.5\textwidth]{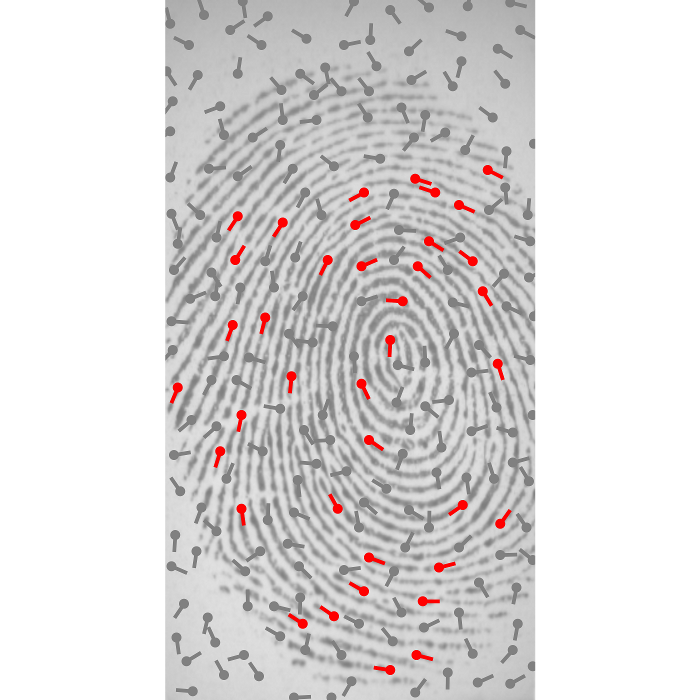}}\subfigure[]{\includegraphics[width=0.5\textwidth]{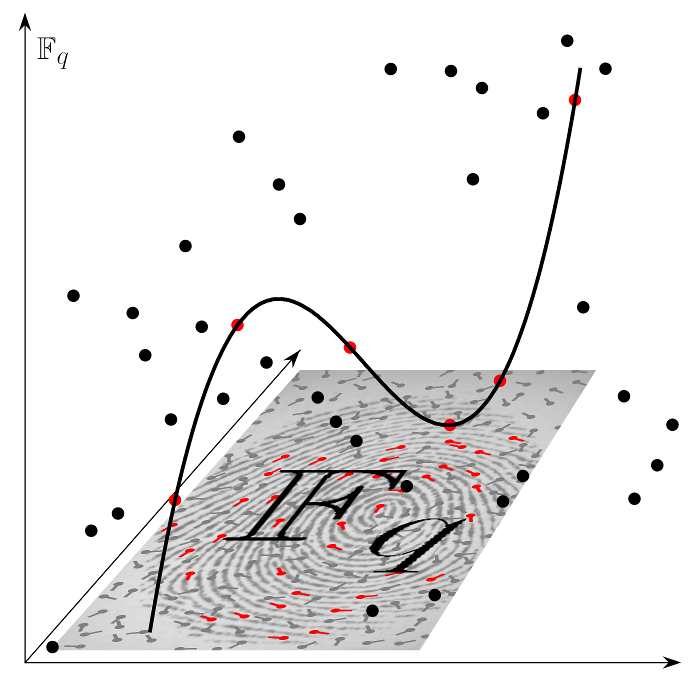}}
\caption{(a) Genuine (red) and chaff minutiae (gray); (b) each minutia is encoded on a vault point's abscissa where its ordinate binds the minutia to the secret polynomial} 
\end{figure}

\subsection{Authentication}
\label{sec:OrdinaryDecoder}
On authentication, a query template of the (alleged) genuine user is provided. As on enrollment, only well-separated minutiae of good quality are selected (again, at most $\tmax$). For simplicity, we assume that the query minutiae are correctly aligned to the vault. We extract those vault minutiae that are well approximated by aligned query minutiae. In this way, we establish the \emph{unlocking set} $\U$ which consists of those vault points that correspond to the just extracted minutiae. Let $t$ be the size of the unlocking set $\U$. There are $\left(t\atop k\right)$ combinations of selecting $k$ different unlocking points. For each combination, the interpolation polynomial $f^*\in\F[X]$ is computed and it is checked whether its hash value agrees with the hash value of the correct polynomial, i.e. if $h(f^*)=h(f)$. If true then $f^*=f$ with overwhelming reliability and $f^*$ is output as the correct polynomial which corresponds to a successful authentication. Otherwise, if all $h(f^*)\neq h(f)$ the authentication attempt is rejected.

\subsection{Alignment}
\label{sec:alignment}
There have been proposals to ease aligning the query minutiae to the vault such that matching vault minutiae agree with their respective query minutiae (see \cite{bib:YangVerbaudwhede2005,bib:JeffersArakala2007,bib:UludagJain2006,bib:NandakumarJainPankanti2007,bib:LiEtAl2008}). All of these proposals leak information about the corresponding fingerprint, e.g., about some of its minutiae or its orientation field. Moreover, it is not clear to what extent auxiliary alignment data can help an adversary to find matching vault correspondences via cross-matching (see Section \ref{sec:CorrelationAttack}).

Ideally, fingerprints can be pre-aligned such that matching query minutiae already agree with genuine vault minutiae. However, pre-alignment is currently not very robust. But increasing robustness of fingerprint pre-alignment would automatically increase the practicability of a minutiae fuzzy vault implementation without decreasing its overall security. Therefore, although challenging, it seems to be worth to search for more robust pre-alignment procedures. Alternatively, suitable alignment-free features can be used for constructing the vault (see \cite{bib:LiEtAl2010}).

Due to open questions related to vault alignment, if we investigate vault performances, we assume a well-solved alignment framework for genuine authentication. Consequently, if an authentication of a genuine user is simulated, the alignment is obtained by aligning the query minutiae template to the enrolled template in clear. On an impostor authentication, we do not make any attempts in aligning the query template to the vault.

\subsection{Evaluation Database and Protocol}
\label{sec:DatabaseAndProtocol}
Throughout, we used the FVC 2002 DB2 database\footnote{The database consists of $8$ impressions each acquired from a total of $100$ fingers.} for our performance evaluations as it is the common database used to evaluate fingerprint fuzzy vault implementations. 

We strictly follow the FVC protocol \cite{bib:FVC2002} even though in the literature the implementations are evaluated following a protocol in where the number of observed impostor recognition attempts is artificially increased \cite{bib:UludagPankantiJain2005,bib:UludagJain2006,bib:NandakumarJainPankanti2007,bib:Nagar2008,bib:Nagar2010,bib:LiEtAl2010}. But this would not correspond to statistically independent observations.

As already described in Section \ref{sec:alignment}, on an genuine authentication attempt we assert that the query finger is correctly aligned by aligning both fingers in clear; for an impostor recognition attempt, we do not make any attempts for alignment.

Genuine acceptance rates and false acceptance rates will be denoted by $\GAR$ and $\FAR$, respectively. Furthermore, throughout the literature the genuine acceptance rate incorporating the first two impression of each fingers only are reported. This corresponds to the scenario in where the fingerprints are of good quality which positively affects the genuine acceptance rates. Therefore, to allow for comparing our genuine acceptance rates with other implementations, we will also keep track of the genuine acceptance rate w.r.t. the subset of the database. The corresponding genuine acceptances rates are indicated by $\subGAR$.

The minutiae templates that we used have been obtained using a commercial extractor.\footnote{Verifinger SDK 5.0 \cite{bib:VERIFINGER}}

\subsection{Performance Evaluation}
\label{sec:MinutiaeFuzzyVaultEvaluation}

\begin{table*}[btp]
\begin{center}
\caption{Performance Evaluation of our minutiae fuzzy vault re-implementation using parameters adopted from Nandakumar et al. (2007) \cite{bib:NandakumarJainPankanti2007}}
\begin{footnotesize}
\label{tab:MinutiaeFuzzyVaultPerformance}
\begin{tabular}{|c||c|c|c|c|}
\hline
                   &                                    &                       &                            &                             \\
 polynomial degree & genuine                            & false                 & avg.                       & avg.                        \\
                   & acceptance                         & acceptance            & genuine                    & impostor                    \\
                   & rate                               & rate                  & decoding                   & decoding                    \\
                   &                                    &                       & time                       & time                         \\
 $<k$              & $\GAR$ ($\subGAR$)                 & $\FAR$                & $\GDT$                     & $\IDT$                      \\
                   &                                    &                       &                            &                             \\
\hline\hline
                   &                                    &                       &                            &                             \\
$=7$               & $\approx 86.54\%$ ($\approx 96\%$) & $\approx 3.87\%$      & $\approx 0.05\seconds$     & $\approx 0.08\seconds$      \\
$=8$               & $\approx 80.76\%$ ($\approx 92\%$) & $\approx 1.63\%$      & $\approx 0.121\seconds$    & $\approx 0.140\seconds$     \\
$=9$               & $\approx 74.61\%$ ($\approx 92\%$) & $\approx 0.56\%$      & $\approx 0.226\seconds$    & $\approx 0.198\seconds$     \\
$=10$              & $\approx 67.52\%$ ($\approx 92\%$) & $\approx 0.16\%$      & $\approx 0.351\seconds$    & $\approx 0.240\seconds$     \\
$=11$              & $\approx 58.93\%$ ($\approx 91\%$) & $\approx 0.10\%$      & $\approx 0.5\seconds$      & $\approx 0.248\seconds$     \\
$=12$              & $\approx 51.07\%$ ($\approx 87\%$) & $=0\%$                & $\approx 0.546\seconds$    & $\approx 0.193\seconds$     \\
                   &                                    &                       &                            &                             \\
\hline
\end{tabular}
\end{footnotesize}
\end{center}
\end{table*}

We evaluated the vault performances for different $k$ on the FVC 2002 DB2 database following the FVC protocol (see \cite{bib:FVC2002}) using parameters adopted from Nandakumar et al. (2007). They propose to hide at most $\tmax=24$ and at least $\tmin=18$ well-separated\footnote{Two minutiae $(a,b,\theta)$ and $(a',b',\theta')$ are said to be well-separated if $\|(a,b)-(a',b')\|_2+0.2\cdot\max(|\theta-\theta'|,|360^\circ-\theta+\theta'|)> 25$.} genuine minutiae in a vault of size $n=224$.
\begin{example}
For example, if $k=9$ then the genuine acceptance rate was determined as $\GAR\approx 74.61\%$ and the false acceptance rate as $\FAR\approx 0.56\%$. The total number of genuine authentication attempts and impostor authentication attempts was $2749$ and $4856$, respectively. $\FTCR\approx 3.88\%$ of the enrollments were aborted, because it was not possible to select at least $\tmin=18$ well-separated minutiae.

The genuine acceptance rate measured on the finger's respective first two impressions was found to be $\subGAR\approx 92\%$ at a failure to capture rate of $\subFTCR=1\%$.  
\end{example}
Note that our rates differ from those reported in \cite{bib:NandakumarJainPankanti2007}. This is due to the use of a different authentication scheme than \cite{bib:NandakumarJainPankanti2007} as well of the fact that we decoupled the alignment from the vault.

We will use the re-implementation to demonstrate the effectiveness of the false-accept attack in Section \ref{sec:FARAttack}. Note, even in the case that we would use auxiliary automatic alignment data to ease alignment, an attacker does not have to account for it; our false acceptance rates reflect the success rate of such a corresponding attack.   

\subsection{Fuzzy Vault with Minutiae Descriptors}
\label{sec:HybridVault}
To improve the practicability as well as the security of the construction in \cite{bib:NandakumarJainPankanti2007}, in addition to mere minutiae, Nagar et al. (2008, 2010) \cite{bib:Nagar2008,bib:Nagar2010} proposed to incorporate minutiae descriptors in constructing the vault. 

A minutia's descriptors consists of the ridge orientation (relative to the orientation of the minutia) and ridge frequency of points arranged around the minutia (see Fig. \ref{fig:MinutiaDescriptor}). The authors showed how a minutia descriptor can be quantized as an $m$-bit vector $w\in\{0,1\}^m$. Furthermore, the corresponding vault points $(x,y)\in\F\times\F$ ordinate value $y$ is encoded as a codeword $c(y)$ of a binary error-correcting code of length $m$ which is capable in correcting $\nu$ (say) errors. The \emph{fuzzy commitment} \cite{bib:JuelsWattenberg1999} of $c(y)$ using the witness $w$ is computed next, i.e. $c(y)+w$.\footnote{Addition is performed bitwise modulo $2$ which is equivalent to a bitwise xor operation.} Rather publishing the vault point $(x,y)$ the tuple $(x,c(y)+w)$ is published instead. For chaff points the ordinate values are protected using random descriptor binarizations from a pool of chaff descriptors.

\begin{figure}[!ht]
\includegraphics[width=0.48\textwidth]{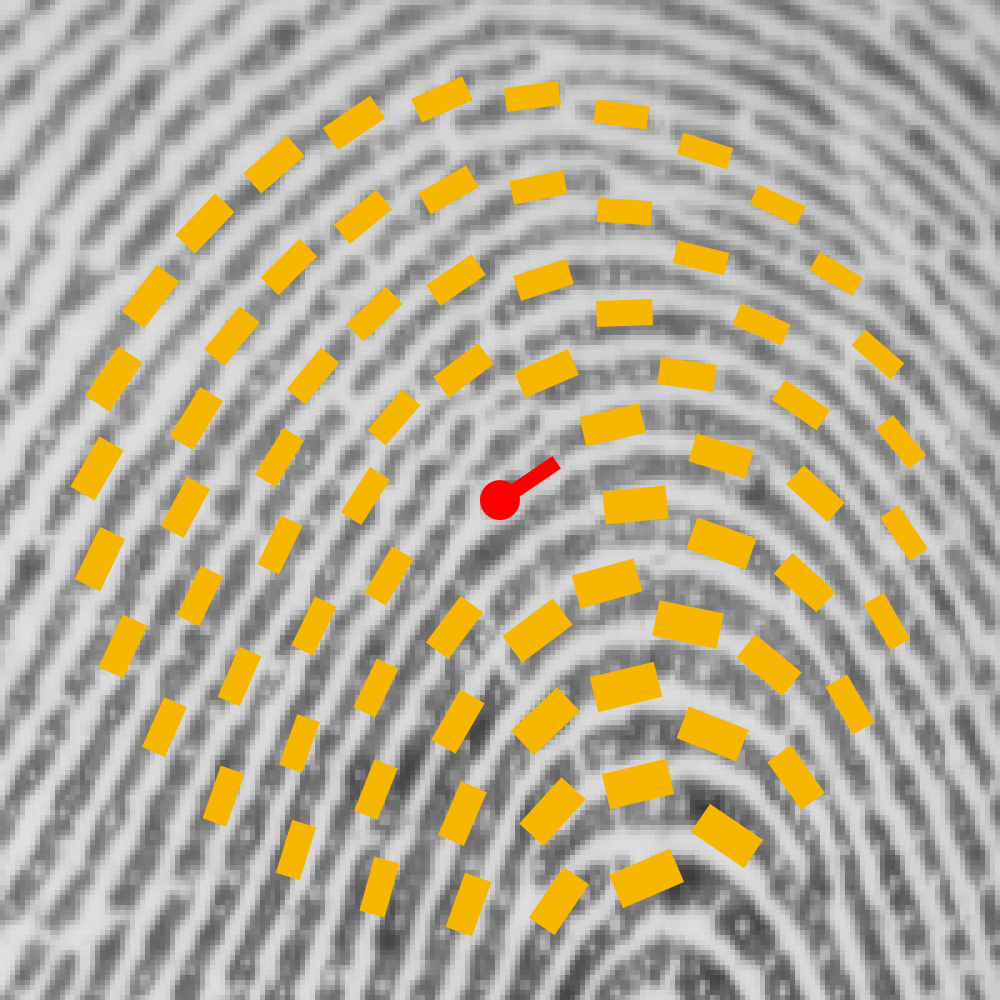}\hspace{0.01\textwidth}\includegraphics[width=0.48\textwidth]{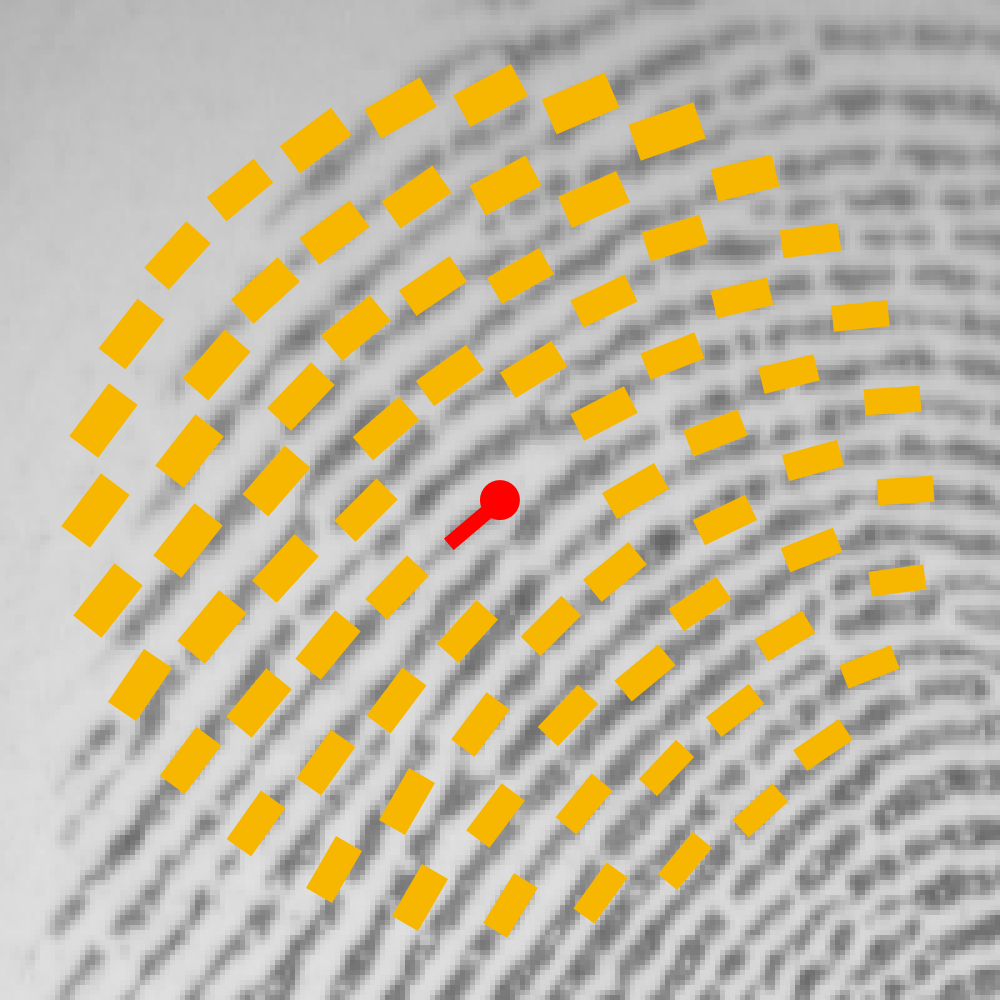}
\caption{\emph{Minutiae descriptors} --- thickness and orientation of yellow lines correspond to ridge frequency and orientation descriptor, respectively; the orientation fields and frequency images to visualize were estimated using the methods in \cite{bib:GottschlichMihailescuMunk2009} and \cite{bib:Gottschlich2012}, respectively.}
\label{fig:MinutiaDescriptor}
\end{figure}

On authentication, the unlocking points $(x,c(y)+w)$ are extracted as in the basic vault. Using the minutiae descriptor $w'\in\{0,1\}^m$ of the corresponding query minutia, the difference $c(y)+w-w'$ is computed.\footnote{Note that addition modulo $2$ is the same as subtraction modulo $2$, i.e. $c(y)+w-w'=c(y)+w+w'$.} If $w'$ is sufficiently similar to $w$, i.e. if they differ in at most $\nu$ positions, the difference $c(y)+w-w'$ can be corrected to $c(y)$ which encodes the correct $y$. Therefore, in addition to sufficiently many genuine vault points among the unlocking set it is required that their correct ordinate values can be recovered. This will be the case, if a minutia descriptor with sufficient similarity to the genuine descriptor can be found. Thus, for the vault to successfully unlock there is required more agreeing information of the query template to the enrolled template and thus the basic vault's security is improved.

Next, we investigate the security of different fingerprint fuzzy vault implementations from the literature.

\section{Attacks}
\label{sec:Attacks}

\subsection{Brute-Force Attack}
\label{sec:BFAttack}

While attack scenarios involving \emph{brute-force attacks} are analyzed throughout the literature they frequently lack of emphasizing how practical these naive attacks can be. In this section, for parameters adopted from implementations found in the literature we determine the expected number of computer time that is expected to be required for a successful brute-force attack. Therefore, we briefly reproduce the work of Mih\u{a}ilescu et al. (2009) \cite{bib:MihailescuMunkTams2009} to emphasize the practicability of brute-force attacks against current implementations of the fuzzy fingerprint vault; for a smart polynomial reconstruction approach we refer to Choi et al. (2011) \cite{bib:ChoiEtAl2011}. Afterwards we modify the attack for the implementation of Nagar et al. (2008, 2010) \cite{bib:Nagar2008,bib:Nagar2010}.

Assume that an intruder has intercepted a vault of size $n$ in which $t$ genuine vault points are contained laying on the graph of a common polynomial of degree $<k$. Furthermore, we assume that a cryptographic hash value $h(f)$ of the correct polynomial is publicly available to the adversary. To find the correct polynomial, the intruder 1) may guess $k$ random vault points, 2) determine its interpolation polynomial $f^*$, and 3) check whether $h(f^*)=h(f)$: If true, the attacker has found the correct polynomial with overwhelming reliability; otherwise, he repeats the attack until a polynomial $f^*$ with $h(f^*)=h(f)$ is found.

The probability that a random choice of $k$ vault points yields the correct polynomial is ${\bf bf}(n,t,k)^{-1}$ where
\begin{equation}
{\bf bf}(n,t,k)=\begin{pmatrix}n\\k\end{pmatrix}\cdot\begin{pmatrix}t\\k\end{pmatrix}^{-1}.
\end{equation}
Thus, after
\begin{equation}
\label{eq:BFSuccessTime}
\log(0.5)/\log(1-{\bf bf}(n,t,k)^{-1})
\end{equation}
iterations the adversary can expect to find the correct polynomial.

\begin{example}
For example, if $(n,t,k)=(224,24,9)$ (which are parameters as proposed by Uludag and Jain 2006 \cite{bib:UludagJain2006}) the adversary can expect to successfully break the vault after $\approx 2^{31}$ iterations using Formula (\ref{eq:BFSuccessTime}). For $\F=\mathbb{F}_{2^{16}}$ we experimentally determined that it is possible to perform $148,634.1$ iterations in one second of the above brute-force attack. Thus, if four processors/cores are used by the attacker he can expect to be successful after $\approx 49\minutes$.
\end{example}

\begin{table}
\begin{center}
\caption{Expected computational timings for running a successful brute-force attack on a 3.2 Ghz desktop computer with four processor cores against different vault parameters found in the literature}
\label{tab:BFAttack}
\begin{tabular}{|c||c|c|c|c|}
\hline
                                       &               &                  &             &                       \\
para-                                  &               &                  & iterations  & {\bf expected}        \\
meters                                 & $(n,t,k)$     & security         & per second  & {\bf time for}        \\
as in                                  &               &                  & per core    & {\bf success}         \\
                                       &               &                  &             &                       \\
\hline\hline
                                       &               &                  &             &                       \\
 \cite{bib:UludagPankantiJain2005}     & $(218,18,9)$  & $\approx 2^{36}$ & $151,316.5$ & $<17\hours$           \\
 \cite{bib:UludagJain2006}             & $(224,24,9)$  & $\approx 2^{31}$ & $148,634.1$ & $<50\minutes$         \\
 \cite{bib:NandakumarJainPankanti2007} & $(224,24,8)$  & $\approx 2^{27}$ & $183,188.8$ & $<3\minutes$          \\
 \verb+"+                              & $(224,24,9)$  & $\approx 2^{31}$ & $148,634.1$ & $<50\minutes$         \\
 \verb+"+                              & $(224,24,11)$ & $\approx 2^{39}$ & $109,066.9$ & $<11\days$            \\
                                       &               &                  &             &                       \\
\hline
                                       &               &                  &             &                       \\
 \cite{bib:LiEtAl2008}                 & $(440,40,13)$ & $\approx 2^{48}$ & $82,056.84$ & $<18\years$           \\
 \verb+"+                              & $(440,40,14)$ & $\approx 2^{52}$ & $69,227.56$ & $<325\years$          \\
                                       &               &                  &             &                       \\
\hline
\end{tabular}
\end{center}
\end{table}

We empirically determined expected times for a successful brute-force attack for parameters from different implementations from the literature. The results are listed in Table \ref{tab:BFAttack}. We find that brute-force attacks can become practical to perform easily --- even on a standard desktop computer. Related work can be found in \cite{bib:MihailescuMunkTams2009,bib:ChoiEtAl2011}.

\subsection{Attack against Fuzzy Vault with Minutiae Descriptors}
To break instances of the implementation of Nagar et al. (2008,2010) \cite{bib:Nagar2008,bib:Nagar2010} the attacker must act differently in choosing a candidate polynomial $f^*$ because the vault points ordinate values are protected. Therefore, we assume that the adversary has access to a large pool of minutiae descriptors.

\subsubsection{Decoupling the Vault from Protected Ordinate Values}
\label{sec:DecoupleMinutiaeDescriptorsFromTheVault}
For each protected vault point $(x,c(y)+w)$ (chaff and genuine) the attacker iterates through the descriptor pool. For each descriptor $w'\in\{0,1\}^m$ the difference $c(y)+w-w'$ is computed and then an attempt is made to decode $c(y)+w-w'$ to its nearest codeword $c(y')$. There are three possible cases:
\begin{itemize}
\item[i)]The attacker can correct to the right $c(y)$ and thus obtains the correct ordinate value $y$;
\item[ii)]he obtains another codeword $c(y')\neq c(y)$ and thus an incorrect ordinate value $y'\neq y$;
\item[iii)]the difference cannot be corrected to any codeword.
\end{itemize}
While iterating through the descriptor pool, the attacker establishes a set of candidate ordinate values. For simplicity, we assume that the correct ordinate value can be found in the candidate set for each vault point. By $\{y'\}$ denote such a candidate set. For the attacker to select a candidate polynomial, he may randomly choose $k$ distinct vault points and, in addition, for each vault point a random candidate ordinate value.

To estimate the probability that such a candidate polynomial yields the correct polynomial, we first estimate the expectation of the size of the candidate set for random vault points $(x,c(y)+w)$. 

An important tool to achieve this is the \emph{sphere packing density} of the underlying binary error-correcting code, which can be defined to be the probability that a random $m$-bit word can be corrected to a valid codeword. Therefore, by $\ell$ denote the number of codewords. Then its sphere packing density is
\begin{equation}
\rho=2^{\ell-m}\sum_{j=0}^\nu\begin{pmatrix}m\\j\end{pmatrix}
\end{equation}
where $\nu$ denotes the code's error-correcting capability. We argue with \cite{bib:Nagar2010} that the difficulty in guessing a random minutiae descriptor is $R\approx 4.27$. Therefore, we estimate the expectation of the number of of candidate ordinate values for each vault point as $S=1+(R-1)\cdot\rho$. Thus, we estimate the brute-force security as
\begin{equation}
S^k\cdot{\bf bf}(n,t,k).
\end{equation}

\subsubsection{Evaluation of the Attack}
\label{sec:HybridVaultBFAttack}
For the implementation of \cite{bib:Nagar2010} in where the ordinate values are protected via a $(511,19)$-BCH code, which can correct $\nu=119$ errors, the corresponding sphere packing density is $\rho\approx 1.3\cdot 10^{-29}$. Thus, we expect the number of a vault point's candidate ordinate values to be $S=1+(R-1)\cdot\rho\approx 1+4.25\cdot 10^{-27}$. Consequently, the estimated brute-force security for vault parameters $(n,t,k)=(224,24,9)$ is $S^k\cdot{\bf bf}(n,t,k)\approx 2.54\cdot 10^{9}$ which corresponds to $31$ bits. In comparison, the brute-force security for the base implementation without protected ordinate values is ${\bf bf}(n,t,k)\approx 2.54\cdot 10^9$ which is almost the same. Thus, there is virtually no improvement in protecting the vault points ordinate values if the code's sphere packing density is too small. As a countermeasure, the use of two other BCH codes of higher sphere packing density have been investigated in \cite{bib:Nagar2010} with a measurable improvement in the brute-force security. The corresponding evaluations can be found in Table \ref{tab:BFAttackMD}. 

\begin{table}[!ht]
\begin{center}
\caption{Brute-force securities of the implementation of Nagar et al. (2010)  for different choices of BCH codes and for different polynomial degrees. The genuine acceptance rates have been extracted from Figure 7 in \cite{bib:Nagar2010} in where the false acceptance rates have been indicated as to be very close to $0$.}
\label{tab:BFAttackMD}
\begin{tabular}{|c||c|c||c|c||c|c|}
\hline
           & \multicolumn{2}{c||}{}                   & \multicolumn{2}{c||}{}                   & \multicolumn{2}{c|}{}                    \\
           & \multicolumn{2}{c||}{BCH$(511,19)$}      & \multicolumn{2}{c||}{BCH$(31,6)$}        & \multicolumn{2}{c|}{BCH$(15,5)$}         \\
           & \multicolumn{2}{c||}{}                   & \multicolumn{2}{c||}{}                   & \multicolumn{2}{c|}{}                    \\
\hline
           &                  &                       &                  &                       &                  &                       \\
poly-      & brute-           &                       & brute-           &                       & brute-           &                       \\
nomial     & force            & $\operatorname{sub-}$ & force            & $\operatorname{sub-}$ & force            & $\operatorname{sub-}$ \\
degree     & security         & $\GAR$                & security         & $\GAR$                & security         & $\GAR$                \\
           &                  &                       &                  &                       &                  &                       \\
\hline
           &                  &                       &                  &                       &                  &                       \\
 $k=7$     & $\approx 2^{24}$ & $95\%$                & $\approx 2^{27}$ & $94\%$                & $\approx 2^{34}$ & $93\%$                \\
 $k=8$     & $\approx 2^{27}$ & $94\%$                & $\approx 2^{31}$ & $93\%$                & $\approx 2^{40}$ & $93\%$                \\
 $k=9$     & $\approx 2^{31}$ & $93\%$                & $\approx 2^{35}$ & $93\%$                & $\approx 2^{45}$ & $91\%$                \\
 $k=10$    & $\approx 2^{35}$ & $89\%$                & $\approx 2^{39}$ & $87\%$                & $\approx 2^{50}$ & $85\%$                \\
 $k=11$    & $\approx 2^{39}$ & $84\%$                & $\approx 2^{44}$ & $81\%$                & $\approx 2^{56}$ & $76\%$                \\
 $k=12$    & $\approx 2^{43}$ & $78\%$                & $\approx 2^{48}$ & $77\%$                & $\approx 2^{61}$ & $73\%$                \\
           &                  &                       &                  &                       &                  &                       \\
\hline
\end{tabular}
\end{center}
\end{table}

\subsection{False-Accept Attack}
\label{sec:FARAttack}
Brute-force attacks can definitely be improved. For example, one may use statistics of fingerprints to accelerate brute-force attacks. Assuming that the statistics of fingerprints is best reproduced by real fingers, making heuristic considerations one may conclude that an attack that takes advantage out of the system's false-acceptance rate $\FAR$ yields the system's overall security. In any case, such an attack yields an upper bound of the system's overall security and is a hint for the existence of a similar efficient statistical attack.

In the scenario of a false-accept attack we assume that an adversary who has intercepted a vault also has access to a sufficiently large database containing fingerprint templates. 

Then the adversary may try to recover the protected template from the vault off-line by simulating authentication attempts using the templates in the database as the queries. For a random query template, with probability $\FAR$ the vault will unlock and reveal the protected key and template. Thus, the adversary can expect to successfully break the vault after he has simulated $\log(0.5)/\log(1-\FAR)$ authentication attempts. If the average impostor decoding time $\IDT$ is known then the computational cost for a successful false-accept attack can be estimated as
\begin{equation}
\label{eq:FAASuccessTime}
\log(0.5)/\log(1-\FAR)\cdot\IDT.
\end{equation}

\begin{example}
\label{example:FARPointEstimation}
For example, assume that an adversary has intercepted a minutiae fuzzy vault as in Section \ref{sec:MinutiaeFuzzyVault} with $k=9$. With Table \ref{tab:MinutiaeFuzzyVaultPerformance} we may assume that $\FAR\approx 0.56\%$ and $\IDT\approx 0.198\seconds$. Thus, the adversary can expect to successfully break the vault after only $\approx 24.6\seconds$. If four processors/cores are used in parallel the time furthermore reduces to approximately $6.15\seconds$. In comparison to the brute-force attack, which takes $\approx 49\minutes$ on the same computer, the false-accept attack turns out to be the better choice for the adversary and thus poses the more serious risk.
\end{example}

\subsubsection{Confidence of the False Acceptance Rate}
In the above example we assumed that the false acceptance rate was $\FAR\approx 0.56\%$. This is because we observed $27$ false accepts among $4,856$ simulated impostor authentication attempts and thus $\FAR=27/4,856\approx 0.56\%$. But actually, the observation of a false accept is the result of a random sample.

Assume that we observed $s$ false accepts among $N$ impostor recognition attempts. Let $\FAR^*=s/N$ be the \emph{point estimation} for the false acceptance rate. We can only be absolutely certain that $\FAR\in(0\%,100\%)$ but, roughly speaking, it is not very likely that the true false acceptance rate differs from $\FAR^*$ too much. To estimate the confidence of $\FAR^*$, a useful concept is the one of \emph{confidence intervals}.

\begin{defn}[Confidence Interval]
Let $\FAR$ be the system's true (but unknown) false acceptance rate. For a fixed $\gamma\in(0\%,100\%]$ let $\FAR_0\leq\FAR_1$ such that $\FAR^*\in[\FAR_0,\FAR_1]$ for $100\gamma\%$ of all point estimations $\FAR^*$. The interval $[\FAR_0,\FAR_1]$ is called \emph{$\gamma$-confidence interval} for $\FAR$. $\gamma$ is called \emph{confidence level}\footnote{A popular choice for a confidence level is $\gamma=95\%$.} of the interval $[\FAR_0,\FAR_1]$.
\end{defn}

There are methods that compute confidence intervals for a given confidence level $\gamma$ when $s$ false accepts within $N$ impostor recognition attempts have been observed. These are, for instance, the \emph{Clopper-Pearson intervals} \cite{bib:ClopperPearson1934}. 

\begin{example}
\label{example:FAAInterval}
The $95\%$-Clopper-Pearson confidence interval for the false acceptance rate in Example \ref{example:FARPointEstimation} is $[0.36\%,0.81\%]$, i.e. if $s=27$ false accepts among $N=4,856$ impostor recognition attempts have been observed. As a consequence, with a confidence of $95\%$, the expected time needed to perform a successful false-accept attack is between $\approx 4.23\seconds$ and $\approx 9.33\seconds$.
\end{example}

\subsubsection{Rule of Three}
\label{sec:RuleOfThree}
Assume we observed $s=0$ false accepts among $N$ impostor recognition attempts. Even if a point estimation yields a false acceptance rate of $0\%$ this estimation is not very confident. The \emph{rule of three} enables an easy way to estimate a $95\%$-confidence interval in this case (see \cite{bib:HanleyLippman-Hand1983,bib:JocanovicLevy1997}).
\begin{thm}[Rule of Three]
The interval $[0,3/N]$ is a confidence interval of confidence level at least $95\%$.
\end{thm}

\subsubsection{Evaluation}

\begin{table*}[btp]
\begin{center}
\caption{Performance of the false-accept attack against the implementation of Section \ref{sec:MinutiaeFuzzyVault} using the parameters of the performance evaluation. The timings are have been determined on a 3.2\Ghz\ desktop computer with four processor cores.}
\label{tab:FARAttack}
\begin{footnotesize}
\begin{tabular}{|c||c|c|c|c|c|}
\hline
            &             &                 &                          &                                &                        \\
 length     & point       & avg.            & $95\%$-                  & {\bf expected}                 & expected               \\
 of         & estimation  & impostor        & confidence               & {\bf time}                     & time                   \\
 secret     & of the      & decoding        & interval                 & {\bf for a}                    & for a                  \\
 pol.       & false       & time            & for the                  & {\bf successful}               & successful             \\
            & acceptance  &                 & false                    & {\bf false-}                   & brute-                 \\
            & rate        &                 & acceptance               & {\bf accept}                   & force                  \\
            &             &                 & rate                     & {\bf attack}                   & attack                 \\
 $k$        & $\FAR^*$    & $\IDT$          & $[\FAR_0,\FAR_1]$        &                                &                        \\
            &             &                 &                          &                                &                        \\
\hline\hline
            &             &                 &                          &                                &                        \\
$=7$        & $188/4,856$ & $0.08\seconds$  & $[3.33\%,4.45\%]$        & $0.31\seconds$--$0.41\seconds$ & $\approx 11.4\seconds$ \\
$=8$        & $79/4,856$  & $0.140\seconds$ & $[1.29\%,2.02\%]$        & $1.19\seconds$--$1.87\seconds$ & $\approx 2.97\minutes$ \\
$=9$        & $27/4,856$  & $0.198\seconds$ & $[0.37\%,0.81\%]$        & $4.22\seconds$--$9.33\seconds$ & $\approx 49.4\minutes$ \\
$=10$       & $8/4,856$   & $0.240\seconds$ & $[0.07\%,0.32\%]$        & $12.8\seconds$--$58.4\seconds$ & $\approx 13.9\hours$   \\
$=11$       & $5/4,856$   & $0.248\seconds$ & $[0.03\%,0.24\%]$        & $17.9\seconds$--$2.14\minutes$ & $\approx 10.2\days$    \\
$=12$       & $0/4,856$   & $0.193\seconds$ & $[0.00\%,0.06\%]$        & $>54.2\seconds$                & $\approx 6.6\months$   \\
            &             &                 &                          &                                &                        \\
\hline
\end{tabular}
\end{footnotesize}
\end{center}
\end{table*}

\begin{example}
\label{example:FAARuleOfThree}
Assume an adversary has intercepted a minutiae fuzzy vault as in Section \ref{sec:MinutiaeFuzzyVault} with $k=12$. The rule of three states that (with confidence $95\%$) we can only expect the true false-acceptance rate to be $\FAR\approx 0.06\%$. Thus, with an impostor decoding time of $\IDT\approx 0.193\seconds$ using Formula (\ref{eq:FAASuccessTime}) the adversary can expect to successfully break the vault after $\approx 3\minutes$ $37\seconds$. If four processors/cores are used in parallel he may be successful even after $\approx 54.2\seconds$.
\end{example}

For the implementation in Section \ref{sec:MinutiaeFuzzyVault} we estimated the expected computational time of a successful false-accept attack. The way we estimated the expected times is analogous to the estimations in Example \ref{example:FAAInterval} and Example \ref{example:FAARuleOfThree}. The results can be found in Table \ref{tab:FARAttack}. 

\subsubsection{Evaluation against Alignment-Free Fuzzy Fingerprint Vault}
For the alignment-free fuzzy fingerprint vault implementation of Li et al. (2010) \cite{bib:LiEtAl2010} where $(n,t,k)=(440,40,13)$ a false-acceptance rate of $0.04\%$ was reported.\footnote{We refer to Table 3 and 4 in \cite{bib:LiEtAl2010} in where the \emph{sum rule} is used for similarity measurement between vault features and query features.} The authors estimate the false-acceptance rate as a point estimation by observing $N=34,650$ impostor authentication attempts.\footnote{\label{fn:NotStatisticallyIndependent} Actually, these are not statistically independent.} Therefore, we assume that $s=14$ false-accepts were observed in their experiment. The $95\%$-Clopper-Pearson confidence interval for the false-acceptance rate thus is $[0.0221\%,0.0678\%]$. Furthermore, the authors report an average decoding time of $0.192\seconds$.\footnote{The decoding times were reported for genuine authentication attempts only. For simplicity, we assume that it agrees with impostor decoding time.} Consequently, using Formula (\ref{eq:FAASuccessTime}) a false-accept attack may consume between $49\seconds$ and $2.51\minutes$ of computer time if four processors/cores are used. In comparison to the brute-force, which is expected to require $\approx 20\years$ (see Table \ref{tab:BFAttack}), the time for a successful false-accept attack is negligible. Moreover, the time is far away from being acceptable for a secure system.

For $k=14$ no false-accepts were observed by the authors. The rule of three (see Section \ref{sec:RuleOfThree}) states that (with a confidence of $95\%$) the true false acceptance rate is $<0.0087\%$. Assuming $\IDT=0.192\seconds$ we can only expect the false-accept attack to require approximately $6.4\minutes$ which strongly contrasts an alleged security of $52$ bits. 

\subsubsection{Evaluation against Fuzzy Vault with Minutiae Descriptors}

If a vault was intercepted by an intruder in where the ordinate values are protected with minutiae descriptors (see Section \ref{sec:HybridVault}) the false-accept attack can be run without modifications. For example, if $k=12$ using the $(15,5)$-BCH code, no false accepts have been observed within $9,900$ impostor authentication attempt in \cite{bib:Nagar2010}.\textsuperscript{\ref{fn:NotStatisticallyIndependent}} Thus, with the rule of three we can only expect the true false acceptance rate to be $\FAR\leq 0.03\%$. By our experiments (see Table \ref{tab:MinutiaeFuzzyVaultPerformance}) we assume an average impostor decoding time of $\IDT=0.193\seconds$. Consequently, we can only expect a false-accept to last $\log(0.5)/\log(1-\FAR)\cdot\IDT\approx 7\minutes$. If all four processors are used in parallel, the time furthermore reduces to $\approx 2\minutes$.

Let us discuss another interesting point. In \cite{bib:Nagar2010} it is reported that if the $(511,19)$-BCH code is used to protect the ordinate values of the construction of \cite{bib:NandakumarJainPankanti2007} the false acceptance rate drops from $0.7\%$ to $0.01\%$. At a first glance, this may lead to the conclusion that the security is improved by a factor of $\approx 70$. But this is not true: An adversary may decouple the basic vault construction from the protected ordinate values due to a very low sphere packing density. More precisely, the expected number of a protected vault point's ordinate value is estimated as $S=1+4.25\cdot 10^{-29}$ (see Section \ref{sec:HybridVaultBFAttack}). Assuming that each vault point's correct ordinate value is contained in the candidate set, we set $S'=S-1$ as the expected number of wrong ordinate values. Using \emph{Markov's inequality}, the probability that there is at least one wrong candidate is less than $S'=4.25\cdot 10^{-29}$. Thus, the probability that decoupling the protected ordinate values to yield an instance of the basic vault without protected ordinate values is $(1-S')^n=(1-S')^{224}\approx 1-9.52\cdot 10^{-27}$ which is overwhelming. Consequently, if the sphere packing density of the underlying error-correcting code is too small, protecting the ordinate values causes virtually no improvement against the false-accept attack.

\subsection{Intermediate Discussion}
Our investigations clearly show that biometric cryptosystems that are based on a single fingerprint cannot provide sufficient security --- unless the false acceptance is reduced to a cryptographic negligible level: It is very easy to break a single fuzzy fingerprint vault using the false-accept attack. This highly advocates that a secure fingerprint cryptosystem must be based on multiple finger --- or even finger in combination with other biometric modalities.

There remain problems with the fuzzy fingerprint vault that can not be solved merely by switching to multiple fingers and that have to be resolved first. These are the problems of cross-matching and the correlation attack.

\subsection{Cross-Matching and the Correlation Attack}
\label{sec:CorrelationAttack}
One of the most serious risks the fuzzy fingerprint vault is concerned with is its high vulnerability to cross-matching.

\begin{figure}[!ht]
\subfigure[]{\includegraphics[width=0.3\textwidth]{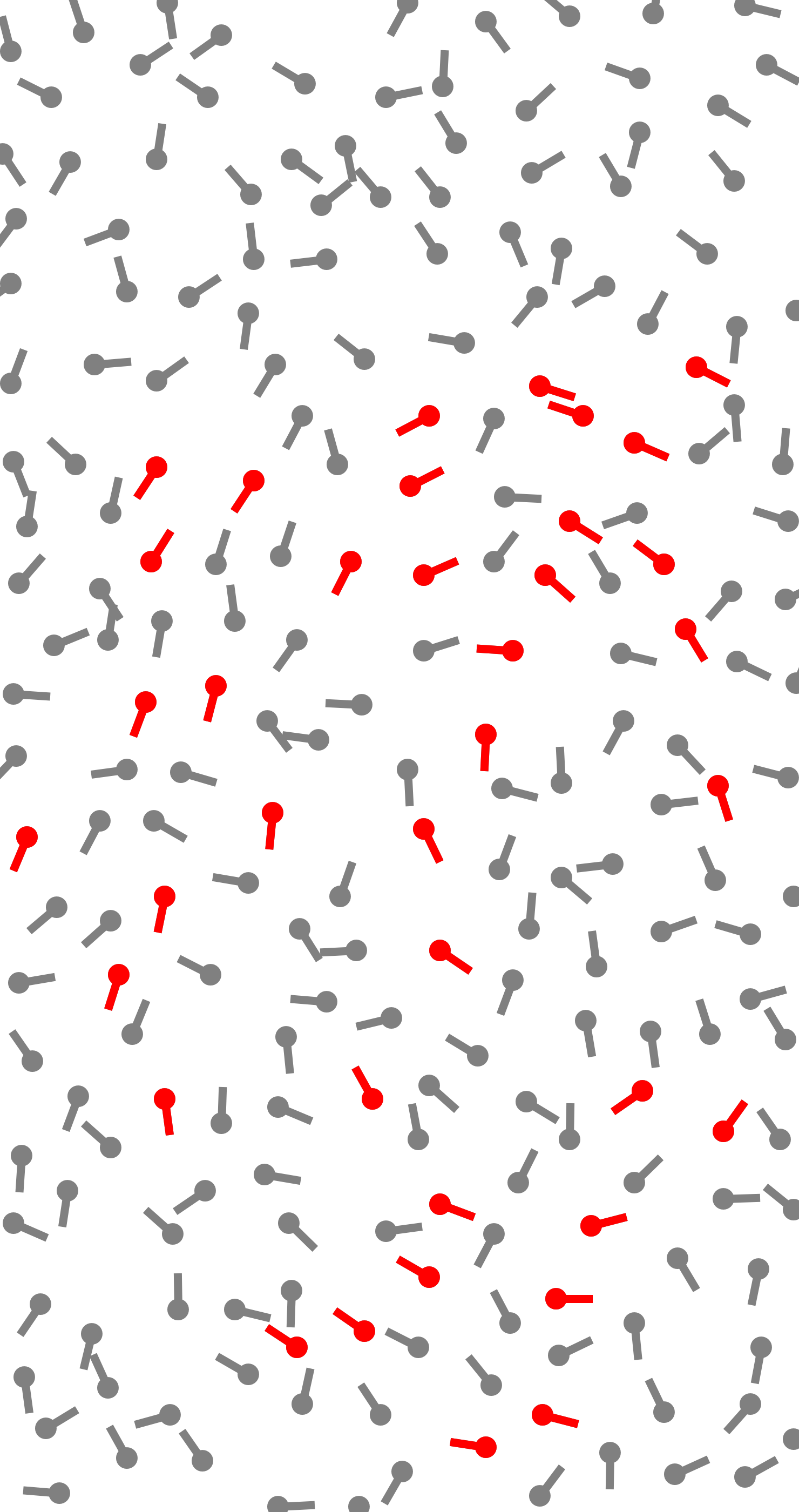}}\hspace{0.025\textwidth}\subfigure[]{\includegraphics[width=0.3\textwidth]{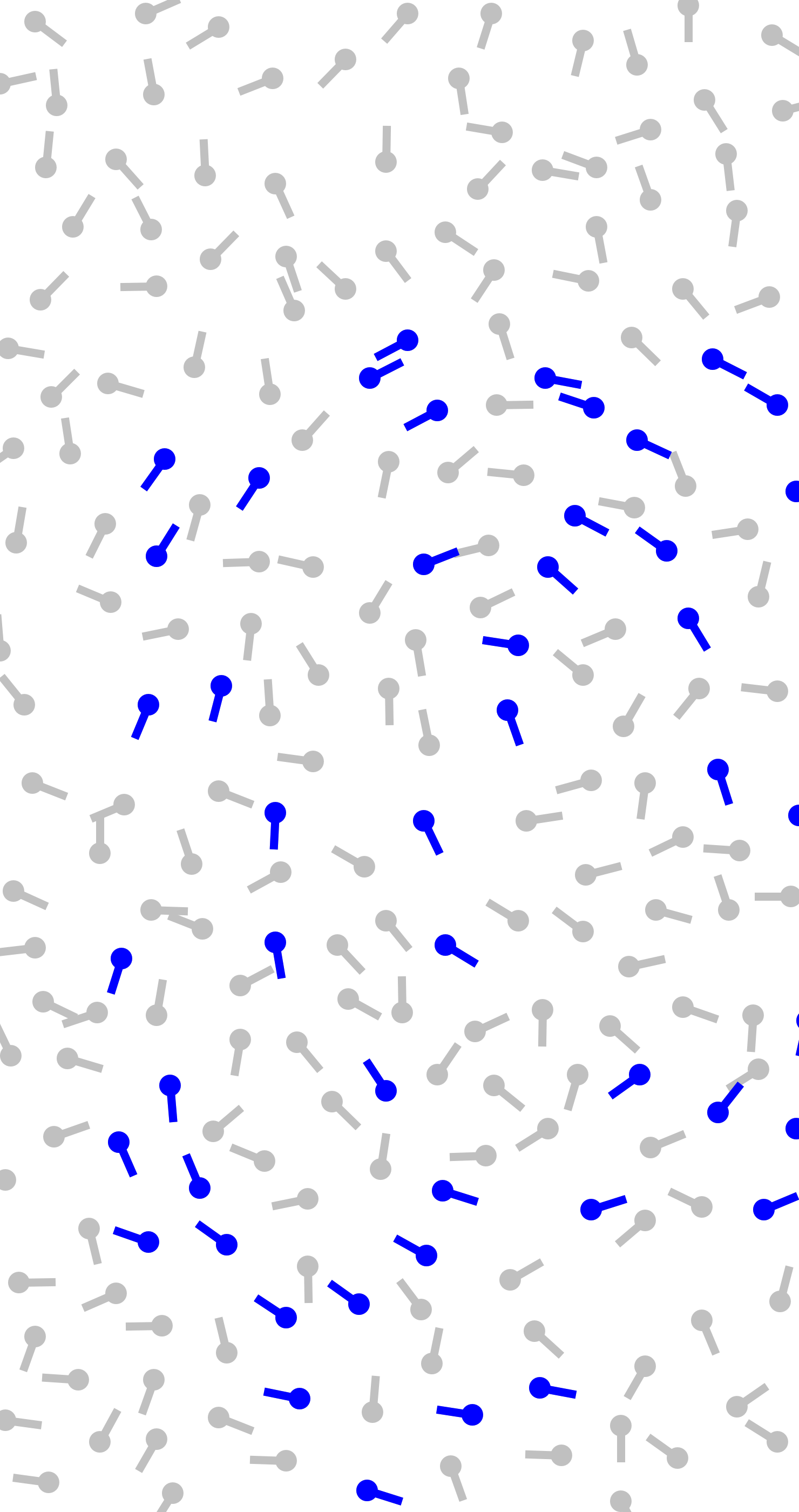}}\hspace{0.025\textwidth}\subfigure[]{\includegraphics[width=0.3\textwidth]{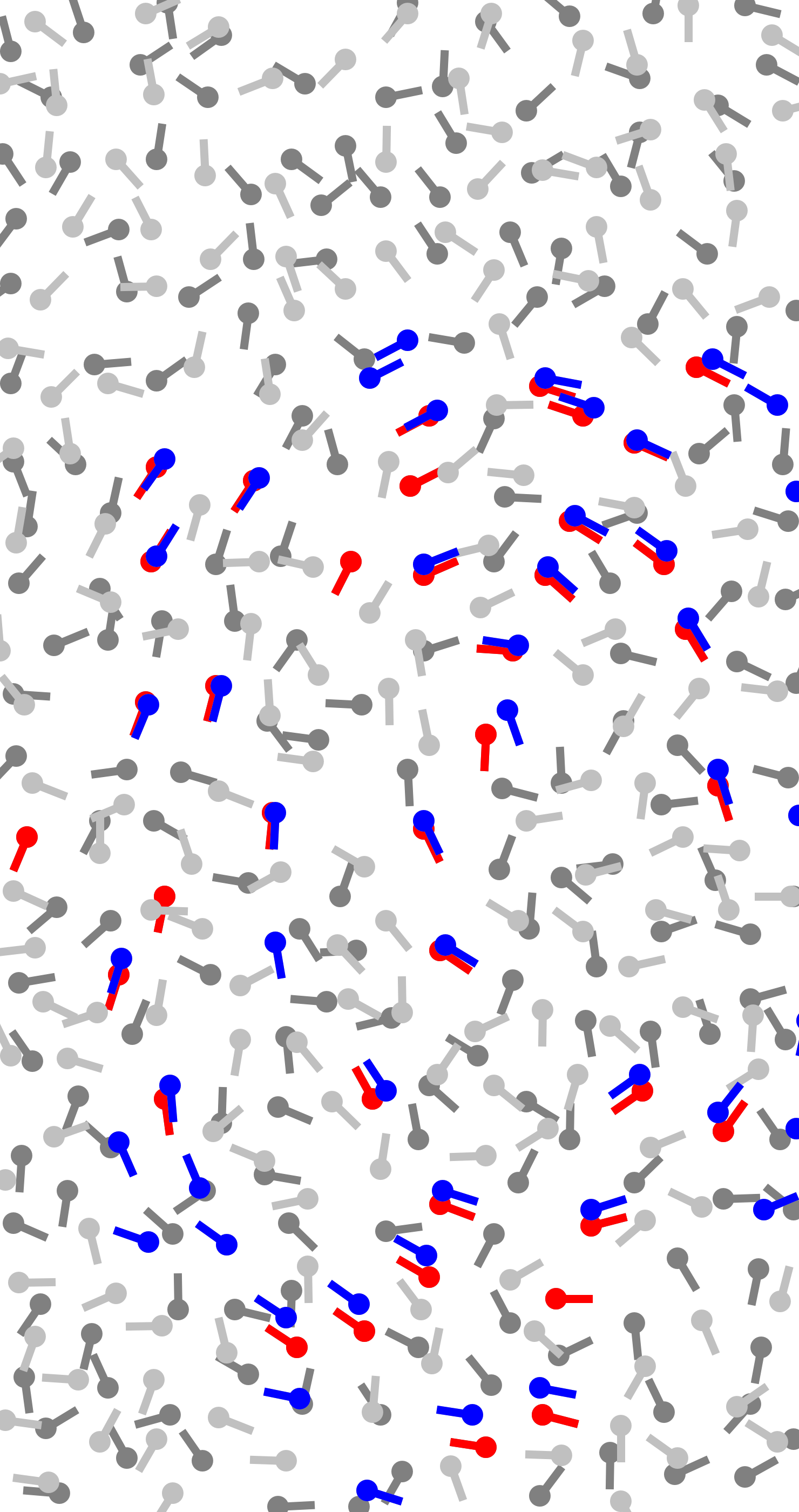}}
\caption{(a), (b) Two aligned vaults with chaff minutiae (gray and light-gray) and genuine minutiae (red and blue). (c) The genuine minutiae have a bias to be in agreement.} 
\label{fig:CorrelationAttack}
\end{figure}

Cross-matching is always possible by, for instance, the brute-force attack: One of the vaults is attacked to reveal its template; this template is then used to open the other vault; if successful, both vaults are considered to match. While such an approach is always possible, there exists a more efficient method to separate genuine points from chaff points, if two vaults protecting the same finger are given: By correlating the vaults, genuine minutiae have a bias to be in agreement in both vaults while chaff minutiae are likely to be separate. An example illustrating this approach is given by Figure \ref{fig:CorrelationAttack}.

While correlation has the inadvertent effect that vault records can be cross-matched, it is even possible for the attacker to efficiently break two vault's templates and keys given the vaults protect templates from the same finger as it was first supposed by Scheirer and Boult in 2007 \cite{bib:ScheirerBoult2007}. As a consequence, separating genuine points from chaff points via correlation has the potential to be much more efficient than merely attacking one of the vaults via brute-force. 

In 2008, Kholmatov and Yanikoglu \cite{bib:KholmatovYanikoglu2008} have demonstrated the effectiveness of the correlation attack. Against the fuzzy vault construction of Uludag et al. (2005) \cite{bib:UludagPankantiJain2005}, they experimentally observed $59\%$ successful recoveries using the correlation attack against $200$ matching vault correspondences. Moreover, the authors were able to perform the correlation attack within $50\seconds$ on average using a non-optimized Matlab implementation on a 3\Ghz\ CPU. In comparison to the brute-force attack, which is expected to last $\approx 80\hours$ on a single core of a 3.2\Ghz\ desktop computer (see Table \ref{tab:BFAttack}), an intruder who has intercepted two matching vault records from different applications may quickly recover the corresponding templates and keys --- even if he has no large database to perform a false-accept attack. 

Cross-matching might already be enabled just by matching alignment helper data (see \cite{bib:LiEtAl2008,bib:UludagJain2006,bib:NandakumarJainPankanti2007}) even though this alone does not imply that multiple records of the same template can be broken efficiently. While the possibility of cross-matching using alignment data alone is already an security issue, it is especially an issue in combination with the correlation attack: An adversary may filter out genuine vault correspondences from different application's databases with the help of the public alignment data; afterwards, he can perform the correlation attack even faster because he can quickly align the vaults using the alignment data. Moreover, merely using alignment-free features as proposed by Li et al. (2010) \cite{bib:LiEtAl2010} will not resolve the risk of cross-matching or attacks via record multiplicity.

Nandakumar, Nagar, and Jain (2008) \cite{bib:NandakumarNagarJain2007} proposed to incorporate an additional user password into the vault. Furthermore, the additional security provided by the user passwords may prevent the vaults from being cross-matched and from being vulnerable to the correlation attack. However, using a user password causes inconveniences that were actually meant to be resolved by biometric based authentication schemes (e.g., weak or forgotten passwords). 

In the next section we show that it is possible to implement a usable fingerprint fuzzy vault that is resistant against the correlation attack and that gets along without an additional user password. 

\section{Implementation of a Cross-Matching Resistant Minutiae Fuzzy Vault}
\label{sec:CMRFV}
We have shown that a single finger is not sufficient to provide a secure biometric cryptosystem due to a cryptographically non-negligible false acceptance rate. Rather biometric cryptosystems that are based on multiple finger/modalities should be developed and analyzed more extensively. First steps have already been made (e.g., see \cite{bib:MerkleEtAl2010b,bib:NagarNandakumarJain2012}). However, it is obvious that merely fusing multiple finger to be protected by the fuzzy vault scheme will not resolve the problem of cross-matching or the correlation attack.

In this section, we propose an implementation of a minutiae fuzzy vault that is inherently resistant against cross-matching and that gets along without an additional password (see \cite{bib:NandakumarNagarJain2007}). Roughly speaking, we achieve cross-matching resistance using the simple idea of rounding minutiae to a rigid hexagonal grid; the minutiae angles are quantized as well. Each element of the rigid system to where a minutia is quantized encodes a genuine vault point while the remaining elements encode chaff points. As a consequence the feature set between different vault records are equal which makes cross-matching via correlation useless to attack the vaults.

\begin{figure}[!ht]
\subfigure[]{\includegraphics[width=0.3\textwidth]{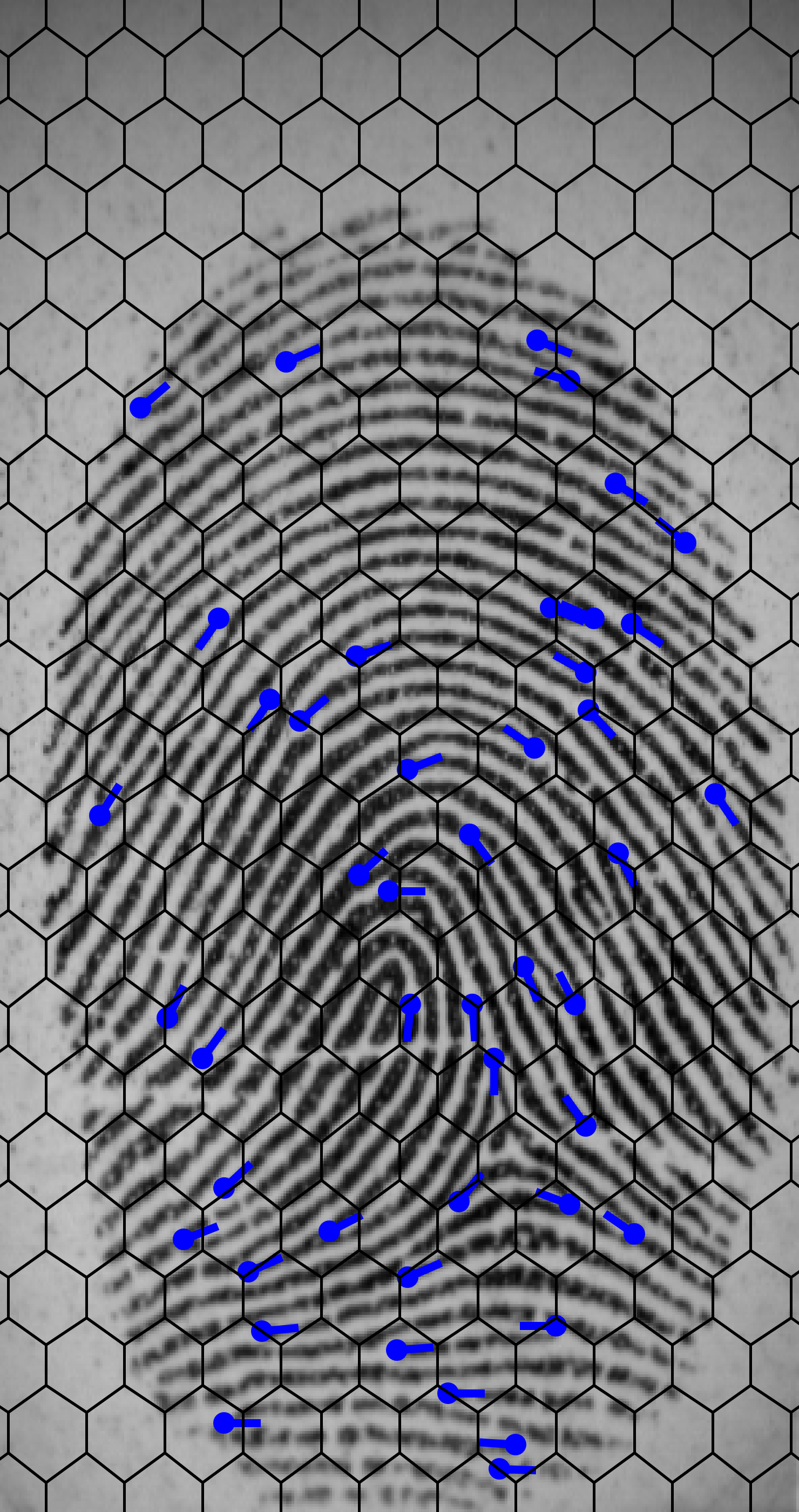}}\hspace{0.025\textwidth}\subfigure[]{\includegraphics[width=0.3\textwidth]{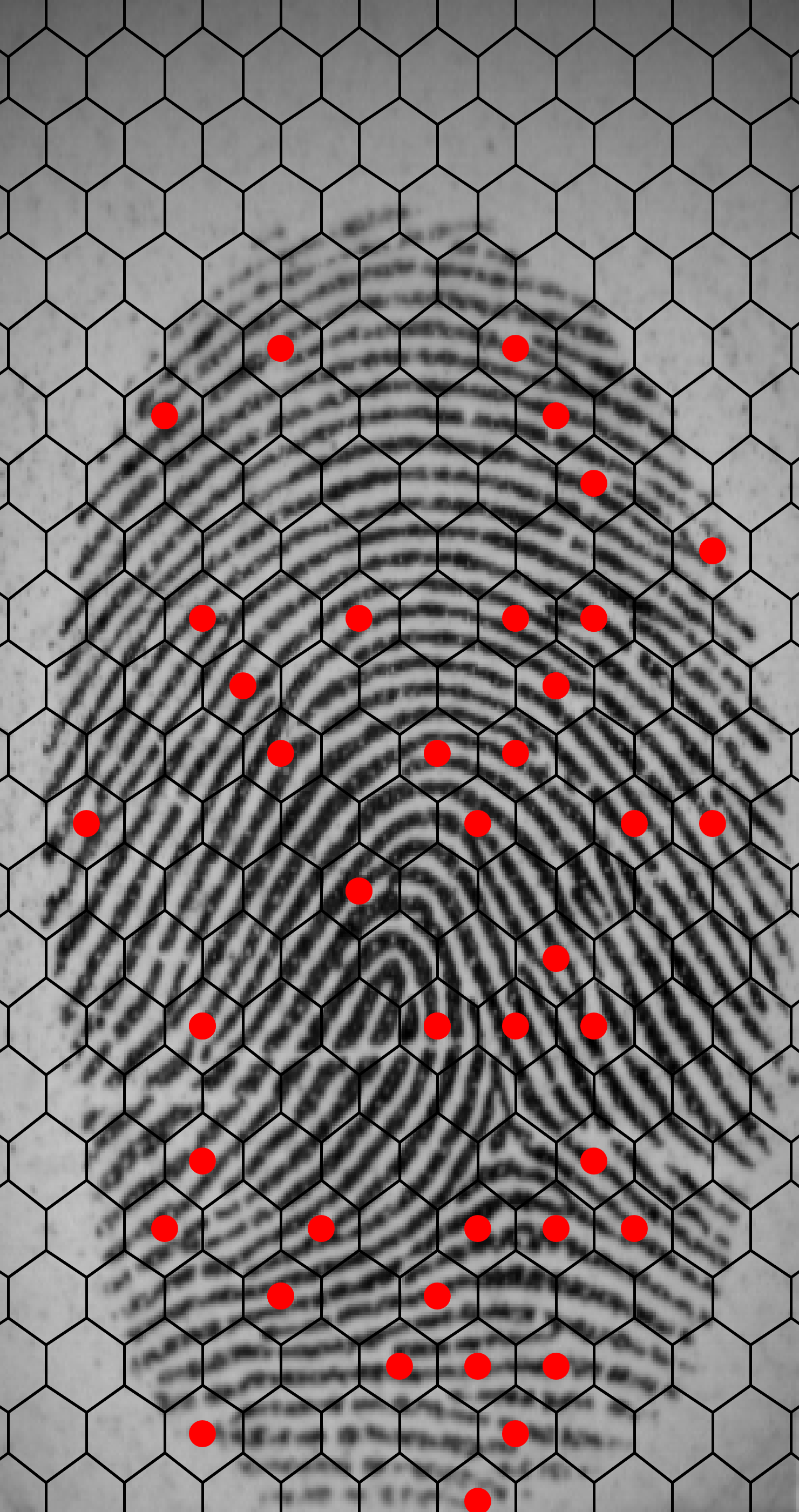}}\hspace{0.025\textwidth}\subfigure[]{\includegraphics[width=0.3\textwidth]{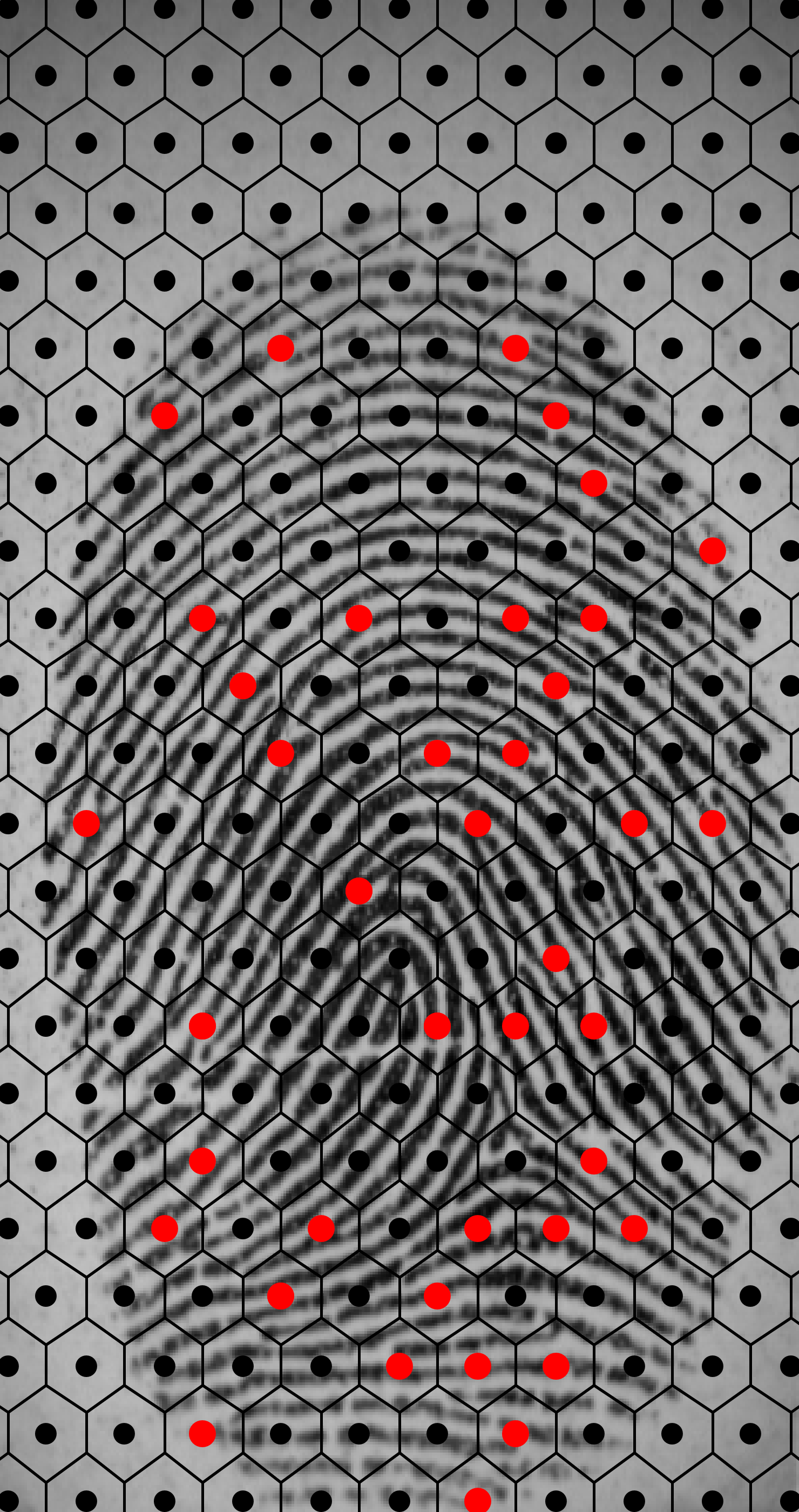}}
\caption{The minutiae (a) positions are rounded to the location of the points of a hexagonal grid (b). Each other point of the grid (c) is used to encode a chaff point.}
\label{fig:HexaVault}
\end{figure}

\subsection{Vault Construction}
\subsubsection*{Minutia Quantization}
Given a minutiae template of a fingerprint, each of its minutia is quantized first. Let $\minutia=(a,b,\theta)$ be a minutia at pixel $(a,b)$ and of angle $\theta\in[0,360)$. Let $R_i$ be the point of a (hexagonal) grid $\{R_0,\hdots,R_{r-1}\}$ laying within the fingerprint image's region that best approximates $(a,b)$. Furthermore, let $j=\lfloor\theta/360\cdot s\rfloor$ where $s$ denotes the parameter controlling the number of values into where angles are quantized. Now, the integer $i+r\cdot j$ encodes the quantization of $\mathfrak{m}$. Let $x_{i,j}\in\F$ denote the finite field element encoding $i+r\cdot j$ by some (but fixed) convention. Then the quantization of the minutia $\mathfrak{m}$ is given by the map $\quant(\mathfrak{m})\mapsto x_{i,j}$.

Note, that the feature set in where minutiae quantizations can occur is ${\bf E}=\{~x_{i,j}~|~i=0,\hdots,r-1,~j=0,\hdots,s-1\}$.

\subsubsection*{Enrollment}
Let $\tmax$ be a bound on the genuine point's size and $\T$ be an input minutiae template that we want to protect. Write $\T=\{\minutia_1,\minutia_2,\hdots\}$ and assume that if $\minutia_i$ is of better quality than $\minutia_j$ this implies $i<j$. The feature set $\A$ is defined to contain at most $\tmax$ quantizations of the first best-quality minutiae. Note, that $\A$ can contain fewer elements than $\T$.\footnote{If there are minutiae in $\T$ that have equal quantization then it is possible that $|\A|<|\T|$.} Let $t=|\A|$. 

The next step is to bind the template quantized as $\A$ to a secret polynomial $f\in\F[X]$ of degree $<k$. This is done as usual by letting the genuine set $\G=\{~(x,f(x))~|~x\in\A~\}$. 

As in the usual vault, the genuine set is hidden among a large set of chaff points. But now, every element in ${\bf E}$ that is not contained in $\A$ corresponds to a chaff point. More precisely, $\C=\{~(x,y)~|~x\in\E\setminus\A~\}$ where the $y$s are chosen uniformly at random from $\F$ such that $y\neq f(x)$.

The vault consists of the union of genuine and chaff points. Furthermore, a cryptographic hash value of the secret polynomial is stored along with the vault. Thus the public vault is the tuple $(\V,h(f))$ where $\V=\G\cup\C$.

\subsubsection*{Vault Authentication}
On authentication, a query minutiae template is given for which we assume that it is already aligned to the vault. Then the corresponding feature set $\B$ is extracted from the query template in the same way as $\A$ was extracted from the enrollment template. Using $\B$ the unlocking set is built out of those points from $\V$ that have abscissa value in $\B$, i.e. $\U=\{~(x,y)\in\V~|~x\in\B~\}$.

Let $\omega=|\A\cap\B|$. Then $\U$ contains exactly $\omega$ genuine points. Thus, if $\omega\geq k$ the secret polynomial $f$ can be obtained from $\U$ in the same way as described in Section \ref{sec:OrdinaryDecoder}.

\subsection{Training}
Our construction is controlled by the following parameters:
\begin{itemize}
\item The minimal distance of the hexagonal grid points $\lambda$ which (together with the fingerprint image's dimension) controls the number of hexagonal grid points $r$;
\item the number of values $s$ into where the minutia's angle are quantized;
\item the bound $\tmax$ on the number of genuine vault points;
\item the size $k$ of the secret polynomial.
\end{itemize}
We performed systematical tests to determine a good configuration of the above parameters. Therefore we determined the $\GAR$s and the $\FAR$s on the FVC 2002 DB2-B (which is intended for training purposes; see \cite{bib:FVC2002}) for each configuration of
\begin{align}
\begin{split}
\lambda=8,\hdots,32;&\quad s=1,\hdots,8;\\
\tmax=10,\hdots,60;&\quad k=1,\hdots,\tmax.
\end{split}
\end{align}
For the $\GAR$s, each finger's $i$th impression was used to extract the feature set $\A$; each $j$th (where $j>i$) impression of the finger aligned to the $i$th was used to extract its features $\B$; if $|\A\cap\B|\geq k$ this was accounted for as a genuine accept; otherwise it was accounted for as a false reject. For the $\FAR$, each $I$th finger's first impression was used to extract the feature set $\A$; each $J$th finger's first impression (where $J>I$) was used to extract $\B$; again, if $|\A\cap\B|\geq k$ this was counted as a false accept; otherwise as a reject.

The best configuration\footnote{The best configuration was defined as to yield the highest $\GAR$ at the lowest $\FAR$; among these configurations, the one with maximal $k/\tmax$ has been selected.} was obtained as 
\begin{align}
\lambda=29,~s=6,~\tmax=44,~\text{and}~k=7
\end{align}
with a $\GAR=100\%$ and $\FAR=0\%$. The number of hexagonal grid points of minimal distance $\lambda=29$ that fit in an image of dimension $296\times 560$ is $r=242$. For the angles are quantized into $s=6$ possible values, the size of the vault is $n=r\cdot s=1452$. Thus, the brute-force security is ${\bf bf}(1452,44,7)\approx 2^{36}$. If a higher security is sought, we may choose a higher $k$.

\subsection{Randomized Decoder}
The potential recognition performance of our construction looks promising. However, there remains a problem concerning the decoding work. If a brute-force security at least $2^{40}$ is sought we may choose $k=8$. On an authentication attempt, an unlocking set of size up to $\tmax=44$ is built. If we would attempt to decode by iterating through all candidate polynomials of degree $<k$ that interpolate $k$ unlocking points, $({44\atop 8})\approx 2^{27}$ iterations have to be performed in the worst case before the user is possibly accepted/rejected. This is too expensive for a usable system.

As a countermeasure, we propose to randomize the decoding procedure of Section \ref{sec:OrdinaryDecoder}. Instead of iterating through all polynomials of degree $<k$ that interpolate $k$ unlocking points, we only iterate through at most $\D$ polynomials each interpolating $k$ randomly selected unlocking points.

On authentication, if the unlocking set $\U$ contains $\omega\geq k$ genuine points the randomized decoder will successfully output the correct polynomial with probability at least $1-(1-{\bf bf}(\tmax,\omega,k)^{-1})^\D$ which approaches $100\%$ as $\D\rightarrow\infty$. Moreover, if the unlocking set $\U$ contains $\omega<k$ genuine points the randomized decoder will not succeed in decoding. Consequently, both $\GAR$ and $\FAR$ for a fixed $k$ drop if the randomized decoder is used.

Furthermore, the use of the randomized decoder only affects authentication and not the vault construction. Thus, for a fixed $k$, the overall security does not suffer if the randomized decoder is used.

\subsection{Performance Evaluation}
\label{sec:CMRFVEvaluation}
\begin{table*}[btp]
\begin{center}
\caption{Result of the Performance Evaluation}
\label{tab:MinutiaeFuzzyExtractorPerfomanceEvaluation}
\begin{footnotesize}
\begin{tabular}{|c||c|c|c|c|c|}
\hline
                   &                             &                       &                            &                             &                       \\
 pol.              & genuine                     & false                 & avg.                       & avg.                        & brute-                \\
 degree            & acceptance                  & acceptance            & genuine                    & impostor                    & force                 \\
                   & rate                        & rate                  & decoding                   & decoding                    & security              \\
                   &                             &                       & time                       & time                        &                       \\
                   &                             &                       &                            &                             &                       \\
     $<k$          & $\GAR$ ($\subGAR$)          & $\FAR$                & $\GDT$                     & $\IDT$                      & ${\bf bf}(1452,44,k)$ \\
                   &                             &                       &                            &                             &                       \\
\hline\hline
                   &                             &                       &                            &                             &                       \\
 $=7$              & $\approx 91.96\%$ ($=97\%$) & $\approx 0.46\%$      & $\approx 0.03\seconds$     & $\approx 0.28\seconds$      & $\approx 2^{36}$      \\
 $=8$              & $\approx 86.82\%$ ($=95\%$) & $\approx 0.081\%$     & $\approx 0.06\seconds$     & $\approx 0.35\seconds$      & $\approx 2^{43}$      \\
 $=9$              & $\approx 80.36\%$ ($=93\%$) & $\approx 0.02\%$      & $\approx 0.1\seconds$      & $\approx 0.41\seconds$      & $\approx 2^{47}$      \\
 $=10$             & $\approx 72.61\%$ ($=91\%$) & $=0\%$                & $\approx 0.17\seconds$     & $\approx 0.51\seconds$      & $\approx 2^{52}$      \\
 $=11$             & $\approx 63.11\%$ ($=86\%$) & $=0\%$                & $\approx 0.26\seconds$     & $\approx 0.60\seconds$      & $\approx 2^{57}$      \\
 $=12$             & $\approx 53.57\%$ ($=80\%$) & $=0\%$                & $\approx 0.39\seconds$     & $\approx 0.73\seconds$      & $\approx 2^{63}$      \\
                   &                             &                       &                            &                             &                       \\
\hline
\end{tabular}
\end{footnotesize}
\end{center}
\end{table*}

For the configuration determined during the training and for different $k$, we performed performance evaluations of our implementation following the description of Section \ref{sec:DatabaseAndProtocol}. We have chosen $\D=2^{16}$ iterations for the randomized decoder which corresponds to a reasonable amount of iterations that is feasible on current hardware. In addition to genuine acceptance rates and false acceptance rates, the average genuine decoding times as well as the average impostor decoding times were determined. The results can be found in Table \ref{tab:MinutiaeFuzzyExtractorPerfomanceEvaluation}. 

In order to compare our results with other fuzzy vault implementations, we also kept track of the genuine acceptance rate in which only the first two impressions are taken into account (the first impression is used for enrollment and the second as the query). The corresponding rates are denoted as $\subGAR$ in Table \ref{tab:MinutiaeFuzzyExtractorPerfomanceEvaluation}. We reached $\subGAR=91\%$ in the case no false accepts have been observed. In comparison, on the same dataset Nagar et al. (2010) \cite{bib:Nagar2010} achieve $\subGAR=93\%$ at zero false accepts while Li et al. (2010) \cite{bib:LiEtAl2010} achieve $\subGAR=92\%$. Even though our results are only valid under a well-solved alignment framework our implementation provides resistance against the correlation attack --- even without a user password (see \cite{bib:NandakumarNagarJain2007}).

\subsection{Alternative Fuzzy Extractor Construction}
Another advantage of our implementation is that it can be easily modified to meet the requirements for the modified fuzzy vault construction proposed by Dodis et al. (2008) \cite{bib:DodisEtAl2008}. This construction avoids the generation of chaff points and significantly reduces the amount of memory that is required for storage. The changes that have to be made would not affect the construction's performance or security against the brute-force or false-accept attack. For details of the construction we refer to \cite{bib:DodisEtAl2008}.

However, without preventions, multiple records of the fuzzy extractor construction may become vulnerable to cross-matching, especially, if the protected templates are equal. The way of cross-matching is similar to the \emph{decodability attack} as it has been investigated by Kelkboom et al. (2011) \cite{bib:KelkboomEtAl2011}. Fortunately, applying a random bit-permutation process secures the fuzzy extractor construction from cross-matching based on the decodability attack. For details we refer to \cite{bib:KelkboomEtAl2011}.

\subsection{Security Analysis}
Our implementation provides good security against the brute-force attack. For example, if $k=10$ at a $52$-bit brute-force security level we have empirically determined that an adversary can test $128,205$ polynomials per second on a single core of a 3.2\Ghz\ desktop computer with four processor cores. Thus, if all four cores are used in parallel, he can expect to break an instance of our implementation after approximately $192\years$.

Our implementation obviously is resistant against the correlation attack and cross-matching via correlation. But the implementation's vulnerability against the false-accept attack remains to be evaluated.

It is possible to analyze the false-accept attack analogous to Section \ref{sec:FARAttack} using confidence intervals. But there is a more elegant way to estimate the false-acceptance rate.

Assume that within an impostor authentication attempt an unlocking set of size $t$ is built containing $\omega$ genuine vault points. Using $\D$ decoding iterations the vault can be unlocked with probability
\begin{equation*}
p(t,\omega,\D)=\begin{cases}
1-(1-{\bf bf}(t,\omega,k)^{-1})^\D,&\text{if}~\omega\geq k\\
0,&\text{if}~\omega<k.
\end{cases}
\end{equation*}
Thus, if in a test with $N$ impostor authentication attempt the $i$th unlocking set is of size $t_i$ containing $\omega_i$ genuine vault points then we may estimate the false acceptance rate as $\FAR\approx\frac{1}{N}\sum_{i=1}^Np(t_i,\omega_i,\D)$. Note, that the effort in authenticating increases linearly with the number of decoding iterations $\D$. Therefore, we estimate the cost for a successful false-accept attack as
$\log(0.5)/\log(1-\FAR)\cdot\D$. 

As the attacker is free in choosing whichever decoder he prefers, he may choose the number of decoding iterations that minimizes the cost.
\begin{lemma}
The cost for a successful false-accept attack is minimized for $\D=1$.
\end{lemma}
\begin{proof}
Let $\epsilon(x)$ be the false-acceptance rate as a function in the number of decoding iterations $x$. Then $g(x)=\log(0.5)/\log(1-\epsilon(x))\cdot x$ is the cost function of a successful false-accept attack. Note that we can write $1-\epsilon(x)=\frac{1}{N}\sum\alpha_i^x$ where $0\leq\alpha_i\leq 1$. Using \emph{Jensen's inequality} we can bound $1-\epsilon(x)\geq\left(\frac{1}{N}\sum\alpha_i\right)^x$. Thus, 
\begin{align*}
g(x)&=\frac{|\log(0.5)|}{|\log(1-\epsilon(x))|}\cdot x\geq\frac{|\log(0.5)|}{\left|\log\left(\left(\frac{1}{N}\sum\alpha_i\right)^x\right)
\right|}\cdot x\\
&=\frac{|\log(0.5)|}{\left|\log\left(\frac{1}{N}\sum\alpha_i\right)\right|\cdot x}\cdot x=\frac{|\log(0.5)|}{\left|\log\left(\frac{1}{N}\sum\alpha_i\right)\right|}=g(1)
\end{align*}
which proves the lemma.
\end{proof}
The lemma enables us to estimate a lower bound for the cost of the false-accept attack against our implementation assuming the adversary also utilizes the randomized decoder. However, the attacker may prefer to use more than only one decoding iteration, e.g., if he uses a fingerprint database for the attack of medium size. Furthermore, the time needed to build the unlocking sets was not taken into account, but it increases the cost for a false-accept attack in practice. But for a security analysis, it is safer to rely on a lower bound.

Hence, in  a test of $N$ impostor authentication attempts in where the $i$th unlocking set was of size $t_i$ containing $\omega_i$ genuine points, the cost for a successful false-accept attack can be estimated as $\log(0.5)/\log(1-\FAR)$ where $\FAR=\frac{1}{N}\sum p(t_i,\omega_i)$ with
\begin{equation*}
p(t_i,\omega_i)=p(t_i,\omega_i,1)=\begin{cases}
{\bf bf}(t_i,\omega_i,k)^{-1}&\text{if}~\omega_i\geq k\\
0&\text{if}~\omega_i<k.
\end{cases}
\end{equation*}

\begin{table}[!ht]
\begin{center}
\caption{Performance of the false-accept attack on a four-core desktop computer with 3.2\Ghz.}
\label{tab:FARAttack2}
\begin{tabular}{|c||c|c|c|}
\hline
                   &                              &                           \\
 polynomial degree & false acceptance rate        & {\bf expected time for a} \\
 $<k$              & $\FAR$                       & {\bf false-accept attack} \\
                   &                              &                           \\
\hline
                   &                              &                           \\
$=7$               & $\approx 8.31\cdot 10^{-8}$  & $\approx 36\seconds$      \\
$=8$               & $\approx 8.87\cdot 10^{-9}$  & $\approx 2\minutes$       \\
$=9$               & $\approx 8.53\cdot 10^{-10}$ & $\approx 21\minutes$      \\
$=10$              & $\approx 6.95\cdot 10^{-11}$ & $\approx 5\hours$         \\
$=11$              & $\approx 4.40\cdot 10^{-12}$ & $\approx 4\days$          \\
$=12$              & $\approx 1.86\cdot 10^{-13}$ & $\approx 120\days$        \\
                   &                              &                           \\           
\hline
\end{tabular}
\end{center}
\end{table}

To estimate the false acceptance rate of our implementation for different $k$, we simulated impostor authentication attempts on the FVC 2002 DB2-A database following the FVC protocol which yielded $N=4,950$ impostor authentication attempts. For the simulations, we used the same configurations as in the performance evaluation. For the $i$th impostor authentication attempt, we quantized the first finger as the set $\A$ and the second as $\B$. Then we let $t_i=|\B|$ and counted $\omega_i=|\A\cap\B|$. Using $p(t_i,\omega_i)$ we estimated the false acceptance rate $\FAR$ for a single decoding iteration and the corresponding cost for the false-accept attack. Consulting the impostor decoding times determined during the evaluation we know how much $2^{16}$ iterations of the false-accept attack cost. We used this information to determine the time for a successful false-accept attack on a 3.2\Ghz\ desktop computer. The results can be found in Table \ref{tab:FARAttack2}. 

\begin{example}
For example, if $k=9$ the false acceptance rate was found to be $\FAR\approx 8.53\cdot 10^{-10}$ for a single decoding iteration. Thus, the attacker can expect to use approximately $\log(0.5)/\log(1-8.53\cdot 10^{-10})\approx 8.13\cdot 10^8$ finger as queries to successfully break the vault. As the time for $2^{16}$ iterations was found to be $\IDT\approx 0.41\seconds$ the time for a successful false-accept attack can be estimated as $8.13\cdot 10^{8}/2^{16}\cdot 0.41\seconds\approx 1$--$2\hours$. If all four processor cores were used in parallel, the time furthermore reduces to $\approx 21\minutes$ which is much more efficient than the brute-force attack requiring $192\years$. Please note, like other implementations our construction is also vulnerable to intensive false-accept attacks.
\end{example}

\section{Discussion and Outlook}
\label{sec:Discussion}
We investigated the security of current implementations of the fuzzy fingerprint vault. We found that, even if the brute-force attack is impractical against some implementations, this does not hold for the false-accept attack. This attack is feasible for every authentication scheme in which the false acceptance rate is non-negligible and thus it is for current implementations of biometric cryptosystem protecting a single fingerprint's template. Even worse, according to our observations, the false-accept attack can be performed much more efficiently than the brute-force attack. One may argue, that it is infeasible for an adversary to establish databases which are of sufficient size to perform intensive false-accept attacks off-line. First, in our view, this can not be prevented having in mind that there exist large databases containing real fingers. Second, the performances of false-accept attacks are hints for the existence of similar efficient statistical attacks. Such attacks may be prevented using multiple fingers or even multiple biometric modalities. Therefore, multi-finger fuzzy vaults should be investigated as a potential method wherever high security is important. And yet a significant risk remains: The correlation attack cannot be prevented merely by using multiple fingers.

Therefore we endeavored to solve the problem of the correlation attack. In this paper we have demonstrated that it is possible to implement a minutiae fuzzy vault that is resistant against the correlation attack without loss of authentication performance. Our implementation primarily relies on the simple innovation of rounding minutiae to a rigid grid while using the entire grid as vault features, thereby preventing attackers from distinguishing genuine from chaff features via correlation. Furthermore, to make vault authentication practical, we proposed to use a randomized decoder rather than systematically iterating through all candidate polynomials. Since the randomized decoder only affects vault authentication and not vault construction, the randomized decoder does not adversely affect vault security. Well conceived, the randomized decoder may be incorporated into a wide variety of fuzzy vault implementations, not only fuzzy vaults with the express purpose of protecting minutiae templates of a single finger. Furthermore, our single-finger fuzzy vault construction that is resistant against the correlation attack may be generalized to a construction that protects multiple fingers. 

All experiments described in this paper can fully be reproduced using software available for download.\footnote{This comprises performance analyses of the brute-force attack, performance evaluations of the minutiae fuzzy vault implementation as in Section \ref{sec:MinutiaeFuzzyVault}, the training for determining a good configuration of our cross-matching resistant minutiae fuzzy vault implementation, its performance evaluations, and a program to analyze our implementation's resistance against the false-accept attack; these are sample programs for a C++ software library that we call \texttt{thimble}; visit \url{http://www.stochastik.math.uni-goettingen.de/biometrics} for downloading its source code.}

We did not propose a mechanism for dealing with alignment for our vault construction. Although it would have been possible to adopt the ideas available in the literature that propose to store additional alignment-helper data publicly with the vault \cite{bib:LiEtAl2008,bib:JeffersArakala2007,bib:NandakumarJainPankanti2007,bib:UludagJain2006,bib:YangVerbaudwhede2005} it is not yet clear how this would affect vault security. Moreover, some of the proposals find accurate alignment via multiple candidate alignments: During authentication, for each candidate alignment an authorization attempt is performed until the correct secret is seen. Translating this method to multiple fingers is problematic because the number of candidate alignments grows exponentially with the number of fingers. Consequently, fingerprint alignment techniques for multi-finger fuzzy vaults should be reconsidered.

Ideally, fingerprints could be pre-aligned. This would make iterations through several candidate alignments obsolete. Moreover, fuzzy vaults protecting accurately pre-aligned fingers do not need to store additional alignment-helper data which can cause unwanted information leakage regarding the corresponding finger. Prealignment of fingerprints is strongly related to the concept of intrinsic coordinate systems \cite{bib:BazenGerez2001b,bib:Hotz2009}. Unfortunately, current methods that extract intrinsic coordinate systems are not robust enough to produce fingerprint pre-alignment of sufficient accuracy. Although challenging, it may be worthwhile to seek more robust methods to extract intrinsic coordinate systems. 

\bibliographystyle{IEEEtran}
\bibliography{IEEEabrv,literature}

\end{document}